%% file: Embedding.tex
\documentclass[11pt, a4paper]{article}%{scrartcl}

\usepackage{amsmath,amsfonts,amssymb,amsthm}
\usepackage{tikz}
\newcommand*\circled[1]{\tikz[baseline=(char.base)]{
            \node[shape=circle,draw,inner sep=2pt] (char) {#1};}}
\usepackage{xspace}
\usepackage{paralist}
\usepackage{graphicx}
\usepackage{array}
\usepackage{fullpage}

\usepackage{url}

\usepackage[ruled,vlined,linesnumbered]{algorithm2e}
\SetFuncSty{textsc}
\SetDataSty{textit}
\SetCommentSty{textrm}
\SetKwFunction{Extract}{ExtractFirstDigit}
\SetKwData{Forward}{forward}
\SetKw{True}{true}
\SetEndCharOfAlgoLine{}

\newcommand{\etal}{\textit{e{}t~a{}l.}\xspace}

\DeclareMathOperator{\poly}{poly}
\DeclareMathOperator{\polylog}{polylog}

\newtheorem{theorem}{Theorem}[section]

\newtheorem{definition}[theorem]{Definition}
\newtheorem{lemma}[theorem]{Lemma}

\newcommand{\seclab}[1]{\label{sec:#1}}
\newcommand{\secref}[1]{Section~\ref{sec:#1}}

\newcommand{\thmlab}[1]{{\label{theo:#1}}}
\newcommand{\thmref}[1]{Theorem~\ref{theo:#1}}

\newcommand{\lemlab}[1]{\label{lemma:#1}}
\newcommand{\lemref}[1]{Lemma~\ref{lemma:#1}}

\newcommand{\figlab}[1]{\label{fig:#1}}
\newcommand{\figref}[1]{Figure~\ref{fig:#1}}

\newcommand{\deflab}[1]{\label{def:#1}}
\newcommand{\defref}[1]{Definition~\ref{def:#1}}

\renewcommand{\O}[1]{\ensuremath{{\mathcal{O}\pth{#1}}}}
\newcommand{\Frechet}{Fr\'echet\xspace}
\newcommand{\atgen}{\symbol{'100}}

\providecommand{\pth}[2][\!]{#1\left({#2}\right)}

\renewcommand{\Re}{{\rm I\!\hspace{-0.025em} R}}
\newcommand{\Na}{{\rm I\!\hspace{-0.025em} N}}

\newcommand{\distFr}[2]{\ensuremath{d_F\pth{#1,#2}}}
\newcommand{\distDTW}[2]{\ensuremath{d_{DTW}\pth{#1,#2}}}

\newcommand{\inner}[2]{\ensuremath{\left\langle #1,#2 \right\rangle}}

\newcommand{\ball}[2]{\ensuremath{\texttt{ball}\pth{#1,#2}}}

\begin{document}

%\linenumbers

\title{Probabilistic embeddings of the \Frechet distance
\thanks{The conference version of this paper will be published at 16th Workshop on Approximation and Online Algorithms (WAOA) 2018.}
}
\author{
   Anne Driemel
   \thanks{Department of Mathematics and Computer Science, TU Eindhoven, The Netherlands; 
     \texttt{a.driemel}\hspace{0cm}\texttt{\atgen{}tue.nl}. 
      Work on this paper was funded by NWO Veni project ``Clustering time series and trajectories (10019853)''. 
} 
   \and
   Amer Krivo\v{s}ija%
   \thanks{Department of Computer Science, TU Dortmund, Germany; 
      \texttt{amer.krivosija}\hspace{0cm}\texttt{\atgen{}tu-dortmund.de}. 
      Work on this paper has been partly supported by DFG within the Collaborative Research Center SFB 876 ``Providing Information by Resource-Constrained Analysis'', project A2.
} 
}

\date{\today}

\maketitle

\begin{abstract}
The Fr\'echet distance is a popular distance measure for curves which naturally lends itself to fundamental computational tasks, such as clustering, nearest-neighbor searching, and spherical range searching in the corresponding metric space. However, its inherent complexity poses considerable computational challenges in practice. To address this problem we study distortion of the probabilistic embedding that results from projecting the curves to a randomly chosen line. Such an embedding could be used in combination with, e.g. locality-sensitive hashing.  We show that in the worst case and under reasonable assumptions, the discrete Fr\'echet distance between two polygonal curves of complexity $t$ in $\Re^d$, where $d\in\lbrace 2,3,4,5\rbrace$, degrades by a factor linear in $t$ with constant probability. We show upper and lower bounds on the distortion. We also evaluate our findings empirically on a benchmark data set. The preliminary experimental results stand in stark contrast with our lower bounds. They indicate that highly distorted projections happen very rarely in practice, and only for strongly conditioned input curves.
\end{abstract}

%\vfill
%\thispagestyle{empty}
%\pagebreak
%\setcounter{page}{1}

%\input{Body}%for the 12 Pages

%\input{BodyWAOA}%for the full version

\input{BodyFull}%for the full version

%\input{Conclusion}

%
\subparagraph*{Acknowledgements.}

We thank Kevin Buchin for useful discussions on the topic of this paper.
%
%\appendix
%\input{Appendix}

\bibliography{Embedding}

%\appendix

%\input{AppendixWAOA}

\end{document}

%% file: BodyFull.tex
%%%Intro.tex

\section{Introduction}

The \Frechet distance is a distance measure for curves which naturally lends
itself to fundamental computational tasks, such as clustering, nearest-neighbor
searching, and spherical range searching in the corresponding metric space. 
However, their inherent complexity poses considerable computational challenges
in practice. Indeed, spherical range searching under the \Frechet distance was
recently the topic of the yearly ACM SIGSPATIAL GISCUP competition\footnote{
6th ACM SIGSPATIAL GISCUP 2017, %see also
\url{http://sigspatial2017.sigspatial.org/giscup2017/} }, highlighting the
relevance and the difficulty of designing efficient data structures for this
problem. At the same time, Afshani and Driemel showed lower bounds on the 
space-query-tradeoff in the pointer model~\cite{ad-ranges-18} that demonstrate that 
this problem is even harder than simplex-range searching.

The computational complexity of computing a single \Frechet distance between
two given curves is a well-studied topic \cite{aaks-dfst-14, Bringmann14, BK-focs15, BringmannK17, buchin2014four, dhw-afd-12, eiter1994computing}. 
It is believed that it takes time that is
quadratic in the length of the curves and this running time can be achieved by
applying dynamic programming.
In this body of literature, the case of 1-dimensional curves under the
continuous \Frechet distance stands out. In particular, no lower bounds are
known on computing the continuous \Frechet distance between 1-dimensional
curves. It has been observed that the problem has a special structure in this
case~\cite{folding-17}.  
Clustering under the \Frechet distance can be done efficiently for
1-dimensional curves \cite{DriemelKS16}, but seems to be harder
for curves in the plane or higher dimensions. 
Bringmann and K\"unnemann used projections to lines to speed up their
approximation algorithm for the \Frechet distance~\cite{ BringmannK17}.
They showed that the distance computation can be done in linear time 
if the convex hulls of the two curves are disjoint.
It is tempting to believe that the curves being
restricted to 1-dimensional space makes the problem significantly easier.
However, in the general case, 
there are no
algorithms known which are faster for 1-dimensional curves than for curves in
higher dimensions.  
In practice, it is very common to separate 
the coordinates
of trajectories to simplify computational tasks. It seems that in practice  
the inherent character of a trajectory is often largely preserved when restricted to
one of the coordinates of the ambient space. Mathematically, this amounts to
projecting the trajectory to a line.

This motivates our study of probabilistic embeddings of the \Frechet distance
into the space of 1-dimensional curves. Concretely, we study distortion of the
probabilistic embedding that results from projecting the curves to a randomly
chosen line.  Such a random projection could be used in combination with
probabilistic data structures, e.g.  locality-sensitive
hashing~\cite{ds-lshc-17}, but also with the multi-level data structures for
\Frechet range searching given by Afshani and Driemel~\cite{ad-ranges-18}. See
below for a more in-depth discussion of these data structures.

We show that in the worst case and under certain assumptions, the discrete
Fr\'echet distance between two polygonal curves of
complexity $t$ in $\Re^d$, where $d=\lbrace 2,3,4,5\rbrace$, degrades by a factor linear in $t$ with constant probability. In
particular, we show upper and lower bounds on the change in distance for the
class of $c$-packed curves. 
The notion of the $c$-packed curves was introduced by Driemel, Har-Peled and
Wenk in \cite{dhw-afd-12} and has proved useful as a realistic input 
assumption~\cite{afpy-dtw-16, Bringmann14, dh-jydfd-13}. 
A curve is called 
$c$-packed for a value $c>0$ if the length of the intersection of the curve with
any ball of any radius $r$ is a most $cr$.
While our study is mostly restricted to the discrete \Frechet distance, we
expect that our techniques can be extended to the case of the continuous
\Frechet distance.

A closely related distance measure, which is popular in the field of data-mining,
is dynamic time warping (DTW)~\cite{dtswk-08,
mueller07dtw,RakthanmanonCMBWZZK12}. The computational complexity of DTW has
also been extensively studied, both empirically and in
theory~\cite{AbboudBW15, afpy-dtw-16, gs-dtw-17, keogh2005exact}. 
Some of our lower bounds extend to DTW.

\subsection{Related work on data structures with \Frechet distance}

The complexity of classic data structuring problems for the \Frechet distance
is still not very well-understood, despite several papers on the topic. We
review what is known for nearest-neighbor searching and range searching.
Indyk \cite{i-approxnn-02} gave a deterministic and approximate near-neighbor
data structure for the discrete \Frechet distance. 
A $c$-approximate nearest-neighbor data structure returns for a given query point $q$ a data point $p\in S$, such that the distance $d(p,q)$ is at most $c\cdot d(p^*,q)$, where $p^*\in S$ is the true nearest neighbor to $q$. Indyk's data structure for data set $S$, containing $n$ curves which have
at most $t$ vertices, achieves approximation factor $c\in\O{\log
t + \log\log n}$ and has query time $\O{\poly(t)\log n }$. This data structure
requires large space, as it precomputes all queries with curves with
$\sqrt{t}$ vertices. For short curves (with $t \in \O{\log n}$) Driemel and
Silvestri \cite{ds-lshc-17} described an approximate near-neighbor 
structure based on locality-sensitive hashing with approximation factor $\O{t}$, 
query time $\O{t \log n}$, using space $\O{n \log n + tn}$. LSH is a technique
that uses families of hash functions with the property that near points are
more likely to be hashed to the same index than far
points. Driemel and Silvestri were the first to define locality-sensitive hash
functions for the discrete \Frechet distance. Emiris and Psarros \cite{EmirisP18} improved their result and also showed how to obtain $(1+\varepsilon)$-approximation with query time $\tilde{\mathcal{O}}\left( d\cdot 2^{2t}\cdot \log n\right)$ and using space $\tilde{\mathcal{O}}(n)\cdot \left( 2+ d/\log t\right)^{\O{t\cdot d\cdot \log(1/\varepsilon)}}$. 
No such hash functions are known
for the continuous case. It is conceivable that the concept of signatures which
was introduced by Driemel, Krivo\v{s}ija and Sohler~\cite{DriemelKS16} in the
context of clustering of 1-dimensional curves could be used to define an LSH
for the continuous case and that this technique could be used in combination
with projections to random lines.

De Berg \etal \cite{BergCG13} studied range counting data structures for
spherical range search queries under the continuous \Frechet distance assuming
that the centers of query ranges are line segments.  This data structure stores compressed
subcurves using a partition tree, using space $\O{s\polylog (n)}$ and query time
$\O{(n/\sqrt{s})\polylog(n)}$ to obtain a constant approximation factor, where
$n\leq s\leq n^2$ is a parameter to the data structure which is fixed at
preprocessing time. 

Afshani and Driemel recently showed how to leverage semi-algebraic range
searching for this problem~\cite{ad-ranges-18}. Their data structure also
supports polygonal curves of low complexity and answers queries exactly. In
particular, for the discrete \Frechet distance they described a data structure
which uses space in $\O{n (\log\log n)^{t_s -1}}$ and achieves query time in
$\O{n^{1-1/d} \cdot \log^{\O{t_s}} n \cdot t_q^{\O{d}}}$, where $t_s$ denotes
the complexity of an input curve and it is assumed that the complexity of the
query curves $t_q$ is upper-bounded by a polynomial of $\log n$.  For the continuous
\Frechet distance they described a data structure for polygonal curves in the
plane which uses space in $\O{n (\log\log n)^{\O{t_s^2}}}$ and achieves query
time in $\O{\sqrt{n} \log^{\O{t_s^2}} n}$. For the case where the curves lie
in dimension higher than $2$ and the distance measure is the continuous
\Frechet distance, no data structures for range searching or range counting are
known.

\subsection{Related work on metric embeddings}

Given metric spaces $(X,d_X)$ and $(Y,d_Y)$, we call a metric embedding an injective mapping $f:X\rightarrow Y$. We call $c$, $c\geq 1$, the distortion of the embedding $f$ \cite{IndMat04} if there is an $r\in (0,\infty)$ such that for all $x,y\in X$ it is $r\cdot d_X(x,y)\leq d_Y(f(x),f(y)) \leq c\cdot r\cdot d_X(x,y)$.

The work that is perhaps closest to ours is a recent result by Backurs and
Sidiropoulos~\cite{approx_BackursS16}. They gave an embedding of the Hausdorff
distance into constant-dimensional $\ell_\infty$ space with constant
distortion. More precisely, for any $s,d\geq 1$, they obtained an embedding for
the Hausdorff distance over point sets of size $s$ in $d$-dimensional space,
into $\ell_{\infty}^{s^{\O{s+d}}}$ with distortion $s^{\O{s+d}}$.
No such metric embeddings are known for the discrete or continuous \Frechet distance.
It has been shown that the doubling dimension of the \Frechet distance is
unbounded, even in the case when the metric spaces is restricted to curves of
constant complexity \cite{DriemelKS16}.  A result of Bartal \etal
\cite{bartal2014impossible} for doubling spaces implies that a metric embedding
of the \Frechet distance into an $\ell_p$ space would have at least
super-constant distortion, but it is not known how to find such an embedding.

We discuss what is known on two variations of the metric embedding problem that are most studied. The first is to find the smallest distortion for any metric from the given class. Matou\v{s}ek \cite{matousek1996distortion} showed that any metric on a point set of size $s$ can be embedded into $d$-dimensional Euclidean space with multiplicative distortion $\O{\min\lbrace s^{2/d}\log^{3/2}s,s\rbrace}$, but not better than $\Omega \left( s^{1/\lfloor (d+1)/2\rfloor}\right)$. For $d=1$ this implies that the distortion is linear in the worst case.

The second problem is to find the smallest approximation factor to a minimal distortion for a given metric over a point set of size $s$. 
We call a spread $\Delta$ a maximum/minimum ratio of the distances of the input point set $X$.
Badoiu \etal \cite{BadoiuCIS05} gave an $\O{\Delta^{3/4}c^{11/4}}$-approximation to the embedding to a line, where $c$ is the distortion of embedding of the input set onto the line. They also showed that it is hard to approximate this problem up to a factor $\Omega\left(n^{1/12}\right)$, even for a weighted tree metrics with polynomial spread.
Assuming a constant distortion $c$ and a polynomial spread $\Delta$, Nayyeri and Raichel \cite{NayyeriR15} gave a $\O{1}$-approximation algorithm to the minimal distortion of the embedding to a line, in time polygonal in $s$ and $\Delta$. See the work of Badoiu \etal \cite{BadoiuDGRRRS05}, Fellows \etal \cite{FellowsFLLRS13}, H{\aa}stad \etal \cite{HastadIL03}, and Indyk \cite{Indyk01} for further reading.

\subsection{Our results}

Given two polygonal curves $P$ and $Q$ with $t$ vertices each from $\mathbb{R}^d$, where $d\in \lbrace 2,3,4,5\rbrace$.
Consider sampling a unit vector $\textbf{u}$ in respective $\mathbb{R}^d$ 
%(resp. $\mathbb{R}^3$ if the curves lie in $\mathbb{R}^3$) 
uniformly at random, and let $P'$ and $Q'$ be the projections of the two curves to the line supporting  $\textbf{u}$. 
We observe that \Frechet distance always decreases when the curves are projected to a line (\lemref{frechetprojection}).
We show that if the curves $P$ and $Q$ are $c$-packed for constant $c$, then, with constant probability, the discrete \Frechet
distance between the curves $P$ and $Q$, denoted by $\distFr{P}{Q}$,  
degrades by at most a linear factor in $t$. This is stated by \thmref{claim:discfrec:n:approx} for $d\in \lbrace 2,3\rbrace$, and by \thmref{claim:discfrec:n:approx45} for $d\in \lbrace 4,5\rbrace$.

\begin{theorem}
\thmlab{claim:discfrec:n:approx}
Given $c\geq 2$, for any two polygonal $c$-packed curves $P$ and $Q$ from $\mathbb{R}^2$ or $\mathbb{R}^3$, and for any $\gamma\in (0,1)$ it holds that
\[
\emph{Pr}\left[ \frac{\distFr{P}{Q}}{\distFr{P'}{Q'}}  \leq \frac{12c+16}{\gamma}\cdot t\right] \geq 1-\gamma.
\]
\end{theorem}

\begin{theorem}
\thmlab{claim:discfrec:n:approx45}
Given $c\geq 2$, for any two polygonal $c$-packed curves $P$ and $Q$ from $\mathbb{R}^4$ or $\mathbb{R}^5$, and for any $\gamma\in (0,1)$ it holds that
\[
\emph{Pr}\left[ \frac{\distFr{P}{Q}}{\distFr{P'}{Q'}}  \leq \left( 1+\frac{2}{\pi}\right) \cdot \frac{12c+16}{\gamma}\cdot t\right] \geq 1-\gamma.
\]
\end{theorem}

We also present a lower bound on the ratio of the two distances. The construction of the lower bound uses 
$c$-packed curves with $c < 3$. 

\begin{theorem}
\thmlab{lowerbounddiscretecpacked}
Given $c\geq 2$, there exist polygonal $c$-packed curves $P$ and $Q$, such that for any $\gamma\in (0,1/\pi)$ 
\[
\emph{Pr}\left[ \frac{\distFr{P}{Q}}{\distFr{P'}{Q'}}  \geq \frac{5\pi\gamma}{6}\cdot t\right] \geq 1-\gamma.
\]
\end{theorem}

\thmref{lowerbounddiscretecpacked} holds for the continuous \Frechet distance
and for dynamic time warping distance as well.

We also show that there exist polygonal curves $P$ and $Q$ that are not
$c$-packed for sublinear $c$ and their (continuous or discrete) \Frechet
distance degrades by a linear factor for any projection line (i.e. with
probability 1).  \thmref{lowerboundcontinuousgeneral} presents this result.

\section{Preliminaries}

Throughout the paper we use the following notational conventions.
Consider two polygonal curves $P=\lbrace p_1, p_2,\ldots, p_t \rbrace$ and $Q=\lbrace
q_1, q_2,\ldots, q_t \rbrace$ in $\mathbb{R}^d$ given by their sequences of vertices.  
We choose a unit vector $\textbf{u}$ in $\mathbb{R}^d$ by choosing a point on the
$(d-1)$-dimensional unit hypersphere uniformly at random. We denote with $L$ the 
line through the origin that supports the vector $\textbf{u}$. 
Let $P'=\lbrace p'_1, p'_2,\ldots, p'_t \rbrace$ and $Q'=\lbrace
q'_1, q'_2,\ldots, q'_t \rbrace$ 
%with $p'_i=\inner{p_i}{\textbf{u}}$ 
be the projections of $P$ and $Q$ to $L$, defined by
$p'_i=\inner{p_i}{\textbf{u}}$ and $q'_j=\inner{q_j}{\textbf{u}}$, for all
$1\leq i\leq t$ and $1\leq j\leq t$.  We denote $\delta_{i,j}=\|p_i-q_j\|$ and
$\delta'_{i,j}=\|p'_i-q'_j\|$, for all $1\leq i\leq t$ and $1\leq j\leq t$,
i.e. $\delta_{i,j}$ and $\delta'_{i,j}$ are the pairwise distances of the
vertices for the input curves $P$ and $Q$ and for their respective projections
$P'$ and $Q'$.

We define the discrete \Frechet distance of $P$ and $Q$ as follows: we call a
\emph{traversal} $T$ of $P$ and $Q$ a sequence of pairs of indices $(i,j)$ of
vertices $(p_i,q_j)\in P\times Q$ such that \begin{compactenum}[i)]
\item the traversal $T$ starts with $(1,1)$ and ends with $(t,t)$, and
\item the pair $(i,j)$ of $T$ can be followed only by one of $(i+1,j)$, $(i,j+1)$ or $(i+1,j+1)$.
\end{compactenum}
We notice that every traversal is monotone. If $\mathcal{T}$ is the set of all
traversals $T$ of $P$ and $Q$, then the discrete \emph{\Frechet distance}
between $P$ and $Q$ is defined as 
\begin{equation}
\label{discretefrechet}
\distFr{P}{Q}=\min_{ T \in \mathcal{T}} \max_{(i,j)\in T} \|p_i-q_j \|.
\end{equation}

Furthermore, we define a directed, vertex-weighted graph $G=(V,E)$ on the node set 
$V=\lbrace (i,j): 1\leq i,j\leq t\rbrace$. A node $(i,j)$ corresponds to a pair of vertices 
$p_i$ of $P$ and $q_j$ of $Q$ and we assign it the weight $\delta_{i,j}$. The set of edges is defined as 
$E=\lbrace \left( (i,j),(i',j') \right): i'\in\lbrace i,i+1\rbrace, j'=\lbrace
j,j+1\rbrace, 1\leq i,i',j,j'\leq t\rbrace$. 
The set of paths in the graph $G$ between $(1,1)$ and $(t,t)$ corresponds to
the set of traversals $\mathcal{T}$. We call a path in $G$ which does not start
in $(1,1)$ or end in $(t,t)$ a \emph{partial traversal} of $P$ and $Q$.

It is useful to picture the nodes of the graph $G$ as a matrix, where rows
correspond to the vertices of $P$ and columns correspond to the vertices of
$Q$. For any fixed value $\Delta>0$,
we define the free-space matrix\footnote{Note that the conventional definition of the free-space matrix for parameter $\Delta$ is slightly different, since usually there is an 1-entry iff $\|p_i-q_j\|\leq \Delta$.  We are using this definition since it better suits our needs. } $F_\Delta=\left( \phi_{i,j}\right)_{1\leq i,j\leq t}$ with
\[
\phi_{i,j}=\begin{cases}
1& \text{if } \|p_i-q_j\|<\Delta\\
0& \text{if } \|p_i-q_j\|\geq\Delta.
\end{cases}
\]

Overlaying the graph with the free-space matrix for $\Delta>\distFr{P}{Q}$, we
can observe that there exists a path in the graph from $(1,1)$ to $(t,t)$ that
visits only the matrix entries with value $1$. Moreover, the existence of such 
a path in the free-space matrix for some value of $\Delta$ implies that $\Delta>\distFr{P}{Q}$.

We define $c$-packedness of curves as follows.

\begin{definition}[$c$-packed curve]\deflab{c_packed_curve}
Given $c>0$, a curve $P\in\mathbb{R}^d$ is $c$-packed if for any point $p\in\mathbb{R}^d$ and any radius $r>0$, the total length of the curve $P$ inside the hypersphere $\ball{p}{r}$ is at most $c\cdot r$.
\end{definition}

%Given curves $P$ and $Q$, their discrete \Frechet distance is realized by some pair $(p_i,q_j)$ of vertices $p_i\in P$ and $q_j\in Q$, being at the distance $\|p_i-q_j\|=\delta$. 

%The proof of the following two lemmas can be found in the appendix.

We prove the following basic fact about random projections to a line, stated for $d\in\lbrace 2, 3\rbrace$ by \lemref{varphibound}, and for $d\in \lbrace 4,5 \rbrace$ by \lemref{varphibound45}. For a general problem in much higher dimension $d$, the probability stated by these lemmas cannot be bounded by a linear function in $\varphi$, due to the measure concentration around $\pi/2$. 

\begin{lemma}
\lemlab{varphibound}
If two points $p$ and $q$ 
%line segment $\overline{pq}$ 
are projected to the straight line $L$, which supports the unit vector chosen uniformly at random on the unit hypersphere in $\mathbb{R}^2$ or $\mathbb{R}^3$, the probability that the distance of their projections will be reduced from the original distance by a factor greater than $\varphi$ is at most $\varphi$.
\end{lemma}

\begin{proof}%[Proof of \lemref{varphibound}]

\begin{figure}[ht]
\centering
\includegraphics[width=0.45\textwidth]{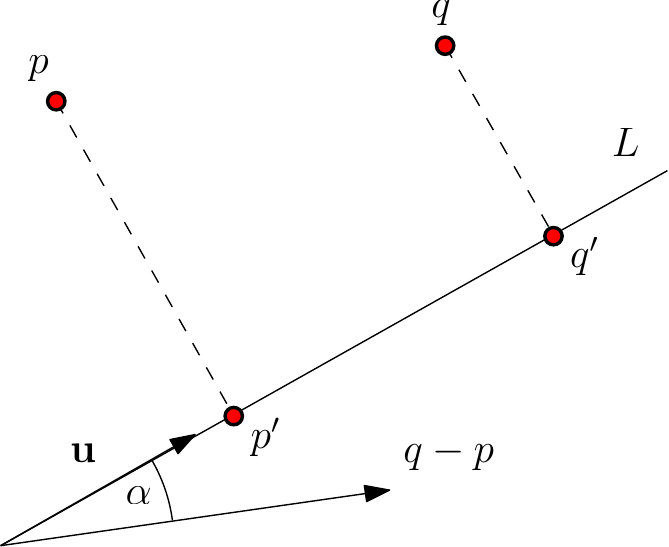}
\caption{The projection of the pair of the vertices to the straight line}
\figlab{projline:case1}
\end{figure}

Let $p$ and $q$ be two vertices in $\mathbb{R}^d$.
Let $\textbf{u}$ be the unit vector chosen uniformly at random on the unit hypersphere and let $L$ be the straight line that supports the vector $\textbf{u}$. Then let $p'$ and $q'$ be the projections of $p$ and $q$ respectively to the projection line $L$, and let $\alpha$ be the angle between $\textbf{u}$ and the vector $q-p$ (see \figref{projline:case1}). Then it holds by the definion of the inner product that
\begin{equation}
\|q'-p'\|= \|\inner{q -p}{\textbf{u}} \cdot \textbf{u} \| = \|q-p\|\cdot \|u\|\cdot |\cos\alpha| \cdot \|u\| .
\label{innerproduct}
\end{equation}
It is $|\cos \alpha|\geq \varphi$ for $\alpha \in \left[ 0, \arccos \varphi \right] \cup \left[ \pi-\arccos\varphi, \pi\right]$, for any $\varphi\in [0,1]$. 

A $d$-sphere is a $d$-dimensional manifold that can be embedded in Euclidean $(d+1)$-dime\-nsional space. A $d$-sphere with radius $R$ has the volume $V_d(R)$ and the surface area $S_d(R)$ given by:
\[
V_d(R)=\frac{\pi^{\frac{d}{2}}}{\Gamma\left(\frac{d}{2}+1\right)}\cdot R^d \hspace{0.5cm}\text{ and }\hspace{0.5cm} S_d(R)=\frac{2\pi^{\frac{d+1}{2}}}{\Gamma\left( \frac{d+1}{2}\right)}\cdot R^d.
\]
$\Gamma(z)$ is the gamma function defined as
\[
\Gamma(z)=\int_0^\infty x^{z-1} e^{-x}dx
\]
for all $z\in\mathbb{R}$, which is a known extension of the factorial function to the set of real numbers, satisfying $\Gamma(1/2)=\sqrt{\pi}$, $\Gamma(1)=1$ and $\Gamma(n+1)=n\cdot\Gamma(n)$ (for all $n\in\mathbb{N}$) \cite{askey2010gamma, huber1982gamma}.

Since the projection line $L$ supports the vector $\textbf{u}$, which is chosen uniformly at random on the unit hypersphere in $\mathbb{R}^d$ ($(d-1)$-sphere with radius 1), the angle $\alpha$ is distributed by the probability distribution function $h_d(\alpha)$, defined as the ratio of the surface of a $(d-2)$-sphere of radius $\sin\alpha$ and the surface of a unit $(d-1)$-sphere. This can be expressed as:
\begin{equation}
\label{eqn:distribution}
h_d\left( \alpha\right)=\frac{1}{\sqrt{\pi}}\cdot \frac{\Gamma\left( \frac{d}{2}\right)}{\Gamma\left( \frac{d-1}{2}\right)}\cdot \left(	\sin \alpha\right)^{d-2}
\end{equation}
over the interval $\alpha\in\left[ 0,\pi\right]$. 
\bigskip

For $d=2$ the distribution of $\alpha$ in Equation (\ref{eqn:distribution}) is uniform with $h_2(\alpha)=1/\pi$. Thus
\begin{equation}
\label{eqn:cosprob}
\text{Pr}\left[ \frac{\|q'-p'\|}{\|q-p\|} \geq \varphi \right] = \text{Pr}\left[ |\cos \alpha |\geq \varphi\right] = \frac{2\arccos \varphi}{\pi}
\end{equation}
and
\[
\text{Pr}\left[ \frac{\|q'-p'\|}{\|q-p\|} < \varphi \right]  = 1- \frac{2\arccos \varphi}{\pi}.
\]
Using Taylor series of $\arccos \varphi$ we get for $0\leq \varphi\leq 1$:
\begin{eqnarray*}
\arccos \varphi &=&  \frac{\pi}{2}- \sum_{k=0}^\infty \frac{(2k)!\cdot \varphi^{2k+1}}{2^{2k}\cdot (2k+1)\cdot (k!)^2}  = \frac{\pi}{2}- \varphi - \sum_{k=1}^\infty \frac{(2k)!\cdot \varphi^{2k+1}}{2^{2k}\cdot (2k+1)\cdot (k!)^2} \\
&\geq& \frac{\pi}{2}-\varphi - \varphi^3\cdot \sum_{k=1}^\infty \frac{(2k)!}{2^{2k}\cdot (2k+1)\cdot (k!)^2} = \frac{\pi}{2}-\varphi-\varphi^3\cdot \left( \frac{\pi}{2}-1\right)
\end{eqnarray*}
since $\varphi\geq \varphi^3 \geq \varphi^{2k+1}$ for all $k\geq 1$. Therefore
\begin{equation}
\label{eqn:phibound2}
\text{Pr}\left[ \frac{\|q'-p'\|}{\|q-p\|} < \varphi \right]  =1-\frac{2\arccos \varphi}{\pi} \leq \frac{2}{\pi}\cdot \varphi + \left( 1-\frac{2}{\pi}\right) \cdot \varphi^3 \leq \varphi.
\end{equation}
\bigskip

For $d=3$ the distribution of $\alpha$ in Equation (\ref{eqn:distribution}) is $h_3(\alpha)=\left(\sin\alpha\right)/2$ for $\alpha\in [0,\pi]$. Thus
\[
\text{Pr}\left[ |\cos \alpha |\geq \varphi\right] = \text{Pr}\left[ \alpha\in [0,\arccos\varphi]\cup [\pi-\arccos\varphi,\pi] \right] = 2\cdot \int_0^{\arccos\varphi} \frac{\sin \alpha}{2}d\alpha = 1-\varphi
\]
due to the symmetry of $h_3(\alpha)$ around $\pi/2$. Therefore it holds that 
\begin{equation}
\label{eqn:phibound3}
\text{Pr}\left[ \frac{\|q'-p'\|}{\|q-p\|} < \varphi \right]  =1- \left( 1-\varphi\right) = \varphi.
\end{equation}

The claim of the lemma follows from Equations (\ref{eqn:phibound2}) and (\ref{eqn:phibound3}). 
%The higher-dimensional cases are discussed by \lemref{varphibound45} in \subsecref{higherdimensionalcases}.
\end{proof}

\begin{lemma}
\lemlab{varphibound45}
If two points $p$ and $q$ 
are projected to the straight line $L$, which supports the unit vector chosen uniformly at random on the unit hypersphere in $\mathbb{R}^4$ or $\mathbb{R}^5$, the probability that the distance of their projections will be reduced from the original distance by a factor greater than $\varphi$ is at most $\left( 1+2/\pi \right)\cdot \varphi$.
\end{lemma}

\begin{proof}%[Proof of \lemref{varphibound45}]

We extend the proof of \lemref{varphibound} to the cases $d=4$ and $d=5$ as follows. 

For $d=4$ the distribution of $\alpha$ in Equation  (\ref{eqn:distribution}) is $h_4(\alpha)=\left(2\sin^2\alpha\right)/\pi$ for $\alpha\in [0,\pi]$. Thus
\[
\text{Pr}\left[ \frac{\|q'-p'\|}{\|q-p\|} < \varphi \right]  =1- 2\cdot\int_0^{\arccos \varphi} \frac{2}{\pi} \sin^2 \alpha d\alpha = 1-\frac{2}{\pi}\left[ \arccos\varphi -\varphi \cdot\sqrt{1-\varphi^2}\right].
\]
Using the last two inequalities of (\ref{eqn:phibound2}), this implies that
\begin{equation}
\label{eqn:phibound4}
\text{Pr}\left[ \frac{\|q'-p'\|}{\|q-p\|} < \varphi \right]  \leq \varphi + \frac{2}{\pi}\cdot\varphi\cdot \sqrt{1-\varphi^2} \leq \left( 1+\frac{2}{\pi}\right) \cdot\varphi.
\end{equation}

\bigskip

For $d=5$ the distribution of $\alpha$ in Equation  (\ref{eqn:distribution}) is $h_5(\alpha)=\left(3\sin^3\alpha\right)/4$ for $\alpha\in [0,\pi]$. Due to the symmetry of $h_5(\alpha)$ around $\pi/2$ it holds that
\begin{equation}
\label{eqn:phibound5}
\text{Pr}\left[ \frac{\|q'-p'\|}{\|q-p\|} < \varphi \right]  =1- 2\cdot\int_0^{\arccos \varphi} \frac{3}{4} \sin^3 \alpha d\alpha = \frac{9}{8}\cdot \varphi - \frac{1}{8} \cos\left( 3\arccos\varphi\right) \leq \left( 1+\frac{2}{\pi}\right) \cdot \varphi.
\end{equation}
The last inequality of (\ref{eqn:phibound5}) follows from the fact that the function $f(\varphi)=(2/\pi-1/8)\cdot \varphi + \cos\left( 3\arccos\varphi\right)/8$ is monotone and increasing, and it holds that $f(0)=0$.

The claim of the lemma follows from Equations (\ref{eqn:phibound4}) and (\ref{eqn:phibound5}). 
%For a general problem in much higher dimension $d$, this probability cannot be bounded by a linear function in $\varphi$, due to the measure concentration around $\pi/2$. 
% For a general problem in $\mathbb{R}^d$ the probability bound degrades due to the measure concentration around $\pi/2$.
\end{proof}

For the sake of completeness we prove the following lemma.

\begin{lemma}
\lemlab{frechetprojection}
Given two curves $P=\lbrace p_1,\ldots,p_t\rbrace$ and $Q=\lbrace q_1,\ldots,q_t\rbrace$ in $\mathbb{R}^d$, and let $P'=\lbrace p'_1,\ldots,p'_t\rbrace$ and $Q'=\lbrace q'_1,\ldots,q'_t\rbrace$ respectively be their projections to the straight line $L$ which supports the vector $\textbf{u}$ chosen uniformly at random on the unit hypersphere in $\mathbb{R}^d$. It holds that $\distFr{P}{Q}\geq\distFr{P'}{Q'}$.
\end{lemma}

\begin{proof}%[Proof of \lemref{frechetprojection}]
Let $\distFr{P}{Q}>\distFr{P'}{Q'}$ for some projection line $L$, and let $T$ and $T'$ be the traversals of $P$ and $Q$, and $P'$ and $Q'$ that realize $\distFr{P}{Q}$ and $\distFr{P'}{Q'}$ respectively. $T,T'\in \mathcal{T}$, where $\mathcal{T}$ is the set of all traversals of $P$ and $Q$ (and also of $P'$ and $Q'$). Then it holds that
\[
\distFr{P}{Q}=\max_{(i,j)\in T} \| p_i-q_j\| < \max_{(i,j)\in T'} \|p'_i-q'_j\| \leq \max_{(i,j)\in T'', T''\in \mathcal{T}} \|p'_i-q'_j\|.
\]
For any $(i,j)\in T''$, $T''\in \mathcal{T}$, we denote with $\alpha_{i,j}$ the angle between the vectors $q_j-p_i$ and $q'_j-p'_i$ (the latter being parallel to $\textbf{u}$). Since any traversal of $P'$ and $Q'$ is a traversal of $P$ and $Q$, using Equation (\ref{innerproduct}) it holds that
\[
\distFr{P}{Q} < \max_{(i,j)\in T} \|p'_i-q'_j\| = \max_{(i,j)\in T} \|p_i-q_j\|\cdot |\cos \alpha_{i,j}| \leq \max_{(i,j)\in T} \|p_i-q_j\|,
\]
a contradiction.
\end{proof}

\section{Upper bound}
\seclab{discreteapproximation}

\subsection{Guarding sets}
\seclab{unionboundlemma}

The discrete \Frechet distance between curves $P$ and $Q$ is realized by some
pair $(p_i,q_j)$ of vertices $p_i\in P$ and $q_j\in Q$, being at the distance
$\|p_i-q_j\|=\delta$.  
We would like to apply  \lemref{varphibound} to this pair of vertices to show
that the distance is preserved up to some constant factor. However, it is possible that the 
pairwise distances in the projection are such that a cheaper traversal is possible 
that avoids the pair $(p_i,q_j)$ altogether.
Therefore, we will apply the lemma to a subset of pairs of vertices of $P$ and $Q$
whose distance is large (e.g. larger than $\Delta=\delta/\theta$ for some small
value of $\theta\geq 1$) and such that the chosen set forms a hitting set for
the set of traversals $\mathcal{T}$.
To this end we introduce the notion of the \emph{guarding set} by the following definition.

\begin{definition}[Guarding set]
\deflab{guardingdiscrete}
For any two polygonal curves $P=\lbrace p_1,\ldots, p_t \rbrace$ and $Q=\lbrace q_1,\ldots, q_t\rbrace$ and a given parameter $\theta\geq 1$, a $\theta$-guarding set $B\subseteq V$ for $P$ and $Q$ is a subset of the set of vertices of $G$ that satisfies the following conditions:
\begin{compactenum}[a)]
\item (distance property) for all $(i,j)\in B$, it holds that $\delta_{i,j} \geq \distFr{P}{Q}/\theta$, and
\item (guarding property) for any traversal $T$ of $P$ and $Q$, it is $T\cap B \neq \emptyset$.
\end{compactenum}
\end{definition}

Note that the set $B$ ``guards'' every traversal of $P$ and $Q$ in the sense
that any path in $G$ from $(1,1)$ to $(t,t)$ has non-empty intersection with
$B$. In other words, $B$ is a hitting set for the set of traversals $\mathcal{T}$. 

For a guarding set $B$ we define the subset of vertices $S_B \subseteq V$ that
can be reached by a path in $G$ starting from $(1,1)$ without visiting a vertex of $B$.
We also define the subset of vertices $H_B= V \setminus (B\cup S_B)$.
A guarding set $B$ thus defines a vertex partition of the graph $G$ into three
subsets $V= S_B \cup B\cup H_B$.

We show the following simple lemma for $d\in \lbrace 2,3\rbrace$, and its counterpart for $d\in \lbrace 4,5\rbrace$, given by \lemref{k:subset:approx45}. 

\begin{lemma}
\lemlab{k:subset:approx}
Given parameter $\theta\geq 1$, if $B$ is a $\theta$-guarding set for the given curves $P=\lbrace p_1,\ldots, p_t \rbrace$ and $Q=\lbrace q_1,\ldots, q_t\rbrace$ from $\mathbb{R}^2$ or $\mathbb{R}^3$, and if $P'$ and $Q'$ are their projections to the straight line $L$, whose support unit vector $\textbf{u}$ is chosen uniformly at random on the unit hypersphere, then for any $\beta>1$ it holds that 
\[
\frac{\distFr{P'}{Q'}}{\distFr{P}{Q}} \geq \frac{1}{\beta\cdot\theta\cdot |B|}
\]
with positive constant probability at least $1-1/\beta$. 
\end{lemma}

\begin{proof}
Let $\textbf{u}$ be the unit vector which is chosen uniformly at random on the unit hypersphere in $\mathbb{R}^d$ with $d\in\lbrace 2,3\rbrace$, and let $\textbf{u}$ be supported by the projection line $L$. Let $\alpha_{i,j}$ be the angle between $\textbf{u}$ and the vector $q_j-p_i$, for $i,j\in\lbrace 1,\ldots,t\rbrace$. If we consider the distances of the pairs of the points $(p_i,q_j)\in P\times Q$, represented by the elements $(i,j)\in B$, then the probability of the event that some of these distances of the points of $P$ and $Q$ is reduced by a factor greater than $\beta\cdot |B|$ (the ``bad'' event) when projected to $L$ can be bounded by the union bound inequality and by \lemref{varphibound} for $\varphi=\frac{1}{\beta |B|}$ as:
\begin{equation}
\label{eqn:unionbound}
\text{Pr}\left[ \left( \exists (i,j)\in B \right) : \frac{\delta'_{i,j}}{\delta_{i,j}} < \frac{1}{\beta|B|} \right] \leq \sum_{(i,j)\in B} \text{Pr}\left[ \frac{\delta'_{i,j}}{\delta_{i,j}} < \frac{1}{\beta|B|} \right]\leq \sum_{(i,j)\in B}\frac{1}{\beta|B|} = \frac{1}{\beta}
\end{equation}
for any $\beta>1$.

Since by \defref{guardingdiscrete} any traversal $T$ of $P$ and $Q$ has nonempty intersection with $B$, the \Frechet distance of $P$ and $Q$ has to be at least as big as the distance of some pair $(i,j)\in T\cap B$. These pairs of vertices have distance at least $\distFr{P}{Q}/\theta$, and they are going to be reduced at most by the factor $\beta\cdot |B|$ (with positive constant probability). The traversal $T'$ of $P'$ and $Q'$ that realizes $\distFr{P'}{Q'}$ has to contain at least one of the pairs of $B$ by \defref{guardingdiscrete}, since the pairs of the traversal $T'$ are simultaneously the pairs of the traversal $T$ of $P$ and $Q$ (that contains the pairs of the vertices of $P$ and $Q$ in the same order as the pairs of their projections in $P'$ and $Q'$). Thus $\distFr{P'}{Q'}\geq \distFr{P}{Q}/\left( \beta \cdot \theta \cdot |B| \right)$, which proves the lemma.
\end{proof}

\begin{lemma}
\lemlab{k:subset:approx45}
Given parameter $\theta\geq 1$, if $B$ is a $\theta$-guarding set for the given curves $P=\lbrace p_1,\ldots, p_t \rbrace$ and $Q=\lbrace q_1,\ldots, q_t\rbrace$ from $\mathbb{R}^4$ or $\mathbb{R}^5$, and if $P'$ and $Q'$ are their projections to the straight line $L$, whose support unit vector $\textbf{u}$ is chosen uniformly at random on the unit hypersphere, then for any $\beta>1$ it holds that 
\[
\frac{\distFr{P'}{Q'}}{\distFr{P}{Q}} \geq \frac{1}{\left( 1+2/\pi \right) \cdot\beta\cdot\theta\cdot |B|}
\]
with positive constant probability at least $1-1/\beta$. 
\end{lemma}

\begin{proof}%[Proof of \lemref{k:subset:approx45}]
We adapt the proof of \lemref{k:subset:approx} as follows: the probability of the ``bad'' event -- that one of the distances of the points of $P$ and $Q$ is reduced by a factor greater than $\left( 1+2/\pi\right)\cdot \beta \cdot |B|$, when projected to $L$, is bounded by the union bound inequality and \lemref{varphibound45} for $\varphi = 1/\left(\left(  1+2/\pi\right)\cdot \beta \cdot |B|\right)$ as:
\begin{eqnarray*}
\label{eqn:unionbound45}
\text{Pr}\left[ \left( \exists (i,j)\in B \right) : \frac{\delta'_{i,j}}{\delta_{i,j}} < \frac{1}{\left(  1+2/\pi\right)\cdot\beta\cdot|B|} \right] & \leq &\sum_{(i,j)\in B} \text{Pr}\left[ \frac{\delta'_{i,j}}{\delta_{i,j}} < \frac{1}{\left(  1+2/\pi\right)\cdot\beta\cdot|B|} \right] \\
& \leq & \sum_{(i,j)\in B}\frac{1+2/\pi}{\left(  1+2/\pi\right)\cdot\beta\cdot|B|} = \frac{1}{\beta},
\end{eqnarray*}
for any $\beta>1$. The rest of the argumentation of the proof of  \lemref{k:subset:approx45} is analogous to the proof of \lemref{k:subset:approx}.
\end{proof}

%\textbf{CHECK}

Intuitively we think of $\delta'_{i,j}$ as an approximation to $\delta_{i,j}$. \lemref{k:subset:approx} yields a naive $\left(\beta\cdot t^2\right)$-approximation 
%to $\distFr{P}{Q}$ 
for any $\beta>1$ and $\theta=1$. Let $B$ be the set of all pairs $(i,j)\in \lbrace 1,\ldots, t\rbrace \times \lbrace 1,\ldots, t\rbrace$ such that $\|p_i-q_j\|=\delta_{i,j} \geq \distFr{P}{Q}$. In the worst case $B$ could contain all $t^2$ pairs. Set $B$ is a $1$-guarding set. The correctness of the condition a) of \defref{guardingdiscrete} is provided directly by the definition of $B$. The condition b) follows by contradiction. If there would exist some traversal $T$ such that $T\cap B=\emptyset$, then for all pairs $(i,j)\in T$ it would have to hold that $\| p_i-q_j\| <\distFr{P}{Q}$. But then the traversal $T$ would witness that $\distFr{P}{Q}\leq \max_{(i,j)\in T} \| p_i-q_j\| < \distFr{P}{Q}$, a contradiction.

One could obtain better constant $\beta$ by more technical argument, which we omit here.
Clearly, the approximation factor  of \lemref{k:subset:approx} can be improved by the better choice of the set $B$. This question we explore in the following section.

\subsection{Improved analysis for c-packed curves}

In order to ensure that the number of the pairs of the indices that take part in the sum in the union bound inequality in (\ref{eqn:unionbound}) is not quadratic but at most a linear one in terms of the input size, we have to carefully select a small subset that satisfies the guarding set properties.

\subsubsection{Building of the initial guarding set}

We first give the simple construction of a $\theta$-guarding set for any $\theta	\geq 1$ by  Algorithm~\ref{guarding:theta}. 

\begin{algorithm}[ht]
 \KwData{$\delta=\distFr{P}{Q}$, vertex-weighted graph $G=(V,E)$}
 \KwResult{set $B$}
\caption{Computing the $\theta$-guarding set, $\theta\geq 1$}\label{guarding:theta}
 	$B\leftarrow \emptyset$\;
	\If{$\delta_{1,1}\geq \delta/\theta$}{
		$B\leftarrow\lbrace (1,1)\rbrace$
		}
	\Else{
		FIFO-Queue $\mathcal{Q}\leftarrow \lbrace (1,1)\rbrace$\tcc*[r]{Breadth-First-Search on $G=(V,E)$}
		\While{$\mathcal{Q}\neq \emptyset$}{

			$(i,j)\leftarrow \text{pop}(\mathcal{Q})$\;
			\ForEach{$((i,j),(i'j'))\in E$}{			
				\If{$\delta_{i,j}<\delta/\theta$ \emph{and} $\delta_{i',j'}<\delta/\theta$}{\text{push}$(\mathcal{Q},(i',j'))$\;
				}
				\ElseIf{$\delta_{i,j}<\delta/\theta$ \emph{and} $\delta_{i',j'}\geq \delta/\theta$}{
					$B\leftarrow B\cup \lbrace (i',j')\rbrace$\;
				}
			}
		}
	\Return $B$
	}
\end{algorithm}

\begin{lemma}
The set $B$ obtained by Algorithm~\ref{guarding:theta} is a $\theta$-guarding set, for any $\theta\geq 1$.
\lemlab{algguardingset}
\end{lemma}
\begin{proof}
We have to show that the resulting set $B$ satisfies the conditions of \defref{guardingdiscrete}. In the case that the distance $\delta_{1,1}\geq \delta/\theta$, it suffices to assign $B=\lbrace (1,1)\rbrace$, since any traversal of the curves $P$ and $Q$ has to include the pair $(1,1)$. For the rest of the proof let $\delta_{1,1}< \delta/\theta$.

Algorithm~\ref{guarding:theta} selects into $B$ only the pairs $(i',j')$ with $\delta_{i',j'} \geq \delta/\theta$ in the line 12, and that are reached by an edge from a pair $(i,j)$ with $\delta_{i,j}\leq \delta/\theta$. Thus the condition a) of \defref{guardingdiscrete} is satisfied by the yielded set. 
For the condition b) we show by induction the following invariant: in each point of time during the BFS, any traversal $T$ contains either a vertex of $B$ or a vertex in the queue $\mathcal{Q}$. The BFS starts with $(1,1)\in \mathcal{Q}$ with $\delta_{1,1}<\delta/\theta$. While processing the pair in $(i,j)\in \mathcal{Q}$ with $\delta_{i,j}<\delta/\theta$ during the BFS (lines 7 and 8) the traversal $T$ may use one of the pairs $(i+1,j)$, $(i,j+1)$ or $(i+1,j+1)$ (connected by the edges in $E$). The next pair in the traversal $T$ is either added into $\mathcal{Q}$ (line 10), or added into $B$ (line 12). In both cases the invariant remains valid. Since the queue is empty at the end, this means that any traversal contains a vertex in $B$, as claimed.
\end{proof}

\begin{figure}[ht]
\centering
\raisebox{-.5\height}{\includegraphics[width=0.5\textwidth]{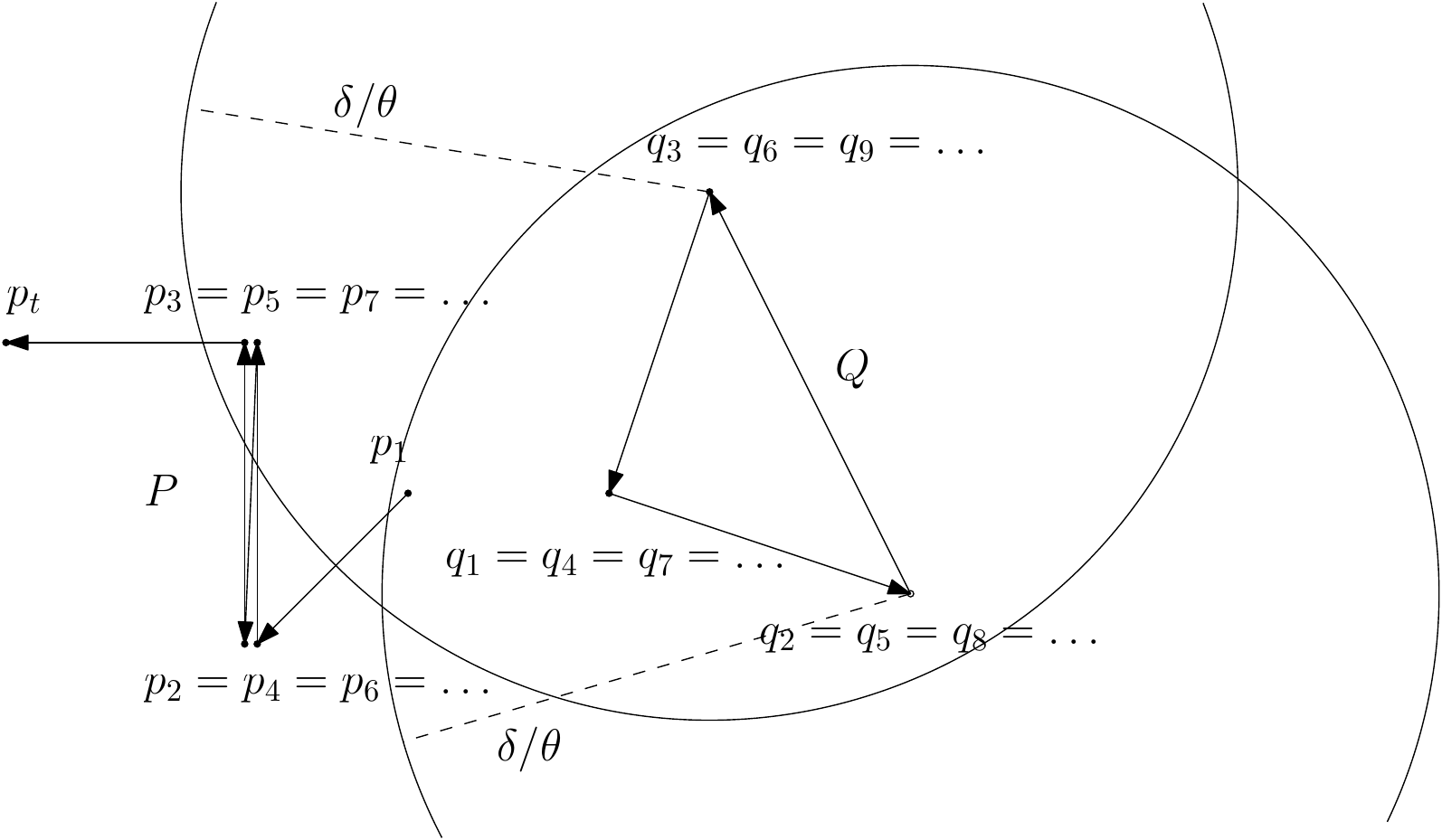}} 
\hspace{0.1cm} 
$F_{\delta/\theta}=\left(\begin{matrix}
\includegraphics[width=0.32\textwidth]{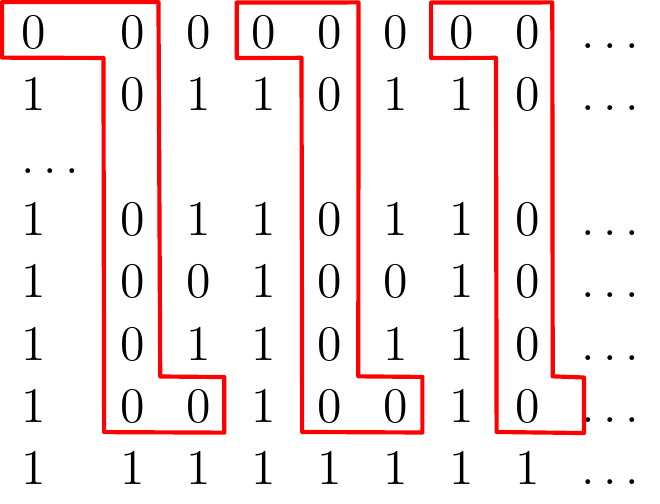}
%0&0&0&0&0&0&0&0&\ldots\\
%1&0&1&1&0&1&1&0&\ldots\\
%\ldots& & & & & & & &\\
%1&0&1&1&0&1&1&0&\ldots\\
%1&0&0&1&0&0&1&0&\ldots\\
%1&0&1&1&0&1&1&0&\ldots\\
%1&0&0&1&0&0&1&0&\ldots\\
%1&1&1&1&1&1&1&1&\ldots
\end{matrix}\right)$
\caption{The curves $P$ and $Q$ (left) that yield a ``fork-like'' free-space matrix $F_{\delta/\theta}$ for some $\theta\geq 1$ (right). The pairs selected into $B$ by Algorithm~\ref{guarding:theta} are marked with the red bound.}
\figlab{forklikecase}
\end{figure}

Unfortunately, the set $B$ built by Algorithm~\ref{guarding:theta} can have a quadratic number of elements in terms of the input size, like the one in \figref{forklikecase} (marked with the red bound). If the free-space matrix $F_{\delta/\theta}$ would have the ``fork-like'' structure for some $\theta\geq 1$, such that for every column $j$ with $j\mod 3=1$ it holds for all pairs $\delta_{i,j}<\delta/\theta$ and thus $\phi_{i,j}=1$ (except for $\delta_{t,j}\geq \delta/\theta$), and for every column $j$ with $j\mod 3=2$ there are all pairs with $\delta_{i,j}\geq \delta/\theta$ and thus $\phi_{i,j}=0$ (except for $\delta_{1,j}<\delta/\theta$). For the columns with $j \mod 3=0$ let $\phi_{1,j}=1$, $\phi_{2,j}=0$ and $\phi_{t,j}=0$ (the rest may be filled arbitrarily).
Then the set $B$ built by Algorithm~\ref{guarding:theta} would contain $\left( t-1\right)\cdot t/3 =\O{t^2}$ entries. We note that this cannot happen if the curves $P$ and $Q$ are $c$-packed for some constant $c$, $c\geq 2$, as it will be discussed in the further text.

\subsubsection{On the structure of the distance matrix}
\seclab{importantpairs}

\lemref{numberofbarriercellsincolumn} states one property of the $c$-packed curves, which we apply in \lemref{comparationofareas}.

\begin{lemma}
\lemlab{numberofbarriercellsincolumn}
Given point $p$ and a $c$-packed curve $Q=\lbrace q_1,\ldots,q_t \rbrace$ from $\mathbb{R}^d$, 
then for any value $b>0$
there exists a value $r \in [b/2,b]$, such that the hypersphere centered at $p$ with radius $r$ intersects 
or is tangent to
at most $2c$ edges of $Q$.
\end{lemma}

\begin{proof}%[Proof of \lemref{numberofbarriercellsincolumn}]
Assume for the sake of contradiction that there exists $c'>2c$, such that for any $r\in [b/2,b]$ there are at least $c'$ edges of $Q$ that intersect or are tangent the surface of the hypersphere $\ball{p}{r}$. Let the \emph{event points} be the points in $\ball{p}{b}\setminus \ball{p}{b/2}$, such that they are either
\begin{compactenum}[i)]
\item vertices $q_i$ of $Q$ or 
\item the points $q'\in \overline{q_i q_{i+1}}$, such that $\overline{pq'}\perp \overline{q_i q_{i+1}}$.
\end{compactenum}
Let the set of events be $\mathcal{R}=\lbrace R_1,\ldots, R_\ell\rbrace$, and let $r_i=\|p-R_i\|$ for all $1\leq i\leq \ell$. We may assume that the events $R_i$ are sorted ascending by $r_i$. Let $r_0=b/2$ and $r_{\ell+1}=b$, thus $r_0\leq r_1\leq \ldots \leq r_{\ell+1}$.

The number of the edges of $Q$ that intersect or are tangent to $\ball{p}{r}$ is equal for all $r'\in \left[ r_i,r_{i+1}\right)$ and for all $0\leq i\leq \ell$, since the number of such edges changes only in event points. After assumption there are at least $c'$ edges of $Q$ that intersect $\ball{p}{r'}$, for any $r'\in \left[ r_i,r_{i+1}\right)$ and for any $0\leq i\leq \ell$. The length of the curve $Q$ within $\ball{p}{b}\setminus \ball{p}{b/2}$ is
\[
\sum_{i=0}^\ell \| Q\cap \left( \ball{p}{r_{i+1}}\setminus \ball{p}{r_i} \right)\| = \| Q\cap \left(\ball{p}{b}\setminus \ball{p}{\frac{b}{2}}\right) \|  \leq c\cdot b
\]
since $Q$ is $c$-packed. But on the other side it is
\[
\sum_{i=0}^\ell \| Q\cap \left( \ball{p}{r_{i+1}}\setminus \ball{p}{r_i} \right)\| \geq \sum_{i=0}^\ell c'\cdot |r_{i+1}-r_i| = c'\cdot \left( b-\frac{b}{2}\right) > c\cdot b,
\]
a contradiction.
\end{proof}

\begin{lemma}
\lemlab{comparationofareas}
Given point $p$ and a $c$-packed curve $Q=\lbrace q_1,\ldots,q_t \rbrace$ from $\mathbb{R}^d$, and given
a value $b>0$, then for any pairwise disjoint set of intervals
\[ I \subseteq \{ [i_1,i_2] ~|~ i_1 \leq i_2 \in \Na, 1 \leq i_1 \leq i_2 \leq t \} \]
with $d(p,q_i) \geq b$ for all $i \in [i_1,i_2] \in I$, there exists a value of $r \in [b/2,b]$ and a 
pairwise disjoint set of intervals 
\[ J \subseteq \{ [j_1,j_2] ~|~ j_1 \leq j_2 \in \Na, 1 \leq j_1 \leq j_2 \leq t \} \]
with the following properties:
\begin{compactenum}[(i)]
\item $|J| \leq c+1$
\item $~\forall~ [j_1,j_2] \in J~ \exists~ i_1 \leq i_2 < i_3 \leq i_4 ~:~ [i_1,i_2], [i_3,i_4] \in I \wedge j_1=i_1 \wedge j_2=i_4 $
\item $~\forall~ i \in [j_1,j_2] \in J ~:~ d(p,q_i) \geq r$
\end{compactenum}
\end{lemma}
\begin{proof} 
We set $r$ to be the value of the same variable as in \lemref{numberofbarriercellsincolumn}.  
Now we construct the set $J$ by merging intervals of $I$ as
follows. Initially $J$ is empty. We iterate over the intervals of $I$ in the order of their starting
points. Consider the first interval $[i_1,i_2]$ and the next interval in the
order $[i_3,i_4]$, we merge them into one interval $[i_1,i_4]$ if there exists
no point $q_j$ with $i_2 < j < i_3$ such that $d(p,q_j) < r$.  We continue
merging this interval with the intervals in $I$ until we found a point $q_j$
such that $d(p,q_j) < r$. Then, we add the current merged interval to $J$ and
take the next interval from $I$ and merge it with the proceeding intervals in
the same manner.  When there are no intervals left in $I$, we also add the
current interval to $J$.  Each time we add an interval to $J$ (except possibly
for the last one), we encountered two edges of $Q$ that intersect the sphere of
radius $r$ centered at $p$.  By \lemref{numberofbarriercellsincolumn} we
have added at most $c+1$ intervals to $J$ (including the last interval).  The
other properties stated in the lemma follow by construction of $J$. \figref{casesthetafig} illustrates the merging process. 
\end{proof}

\begin{figure}[ht]
\centering
\includegraphics[width=0.6\textwidth]{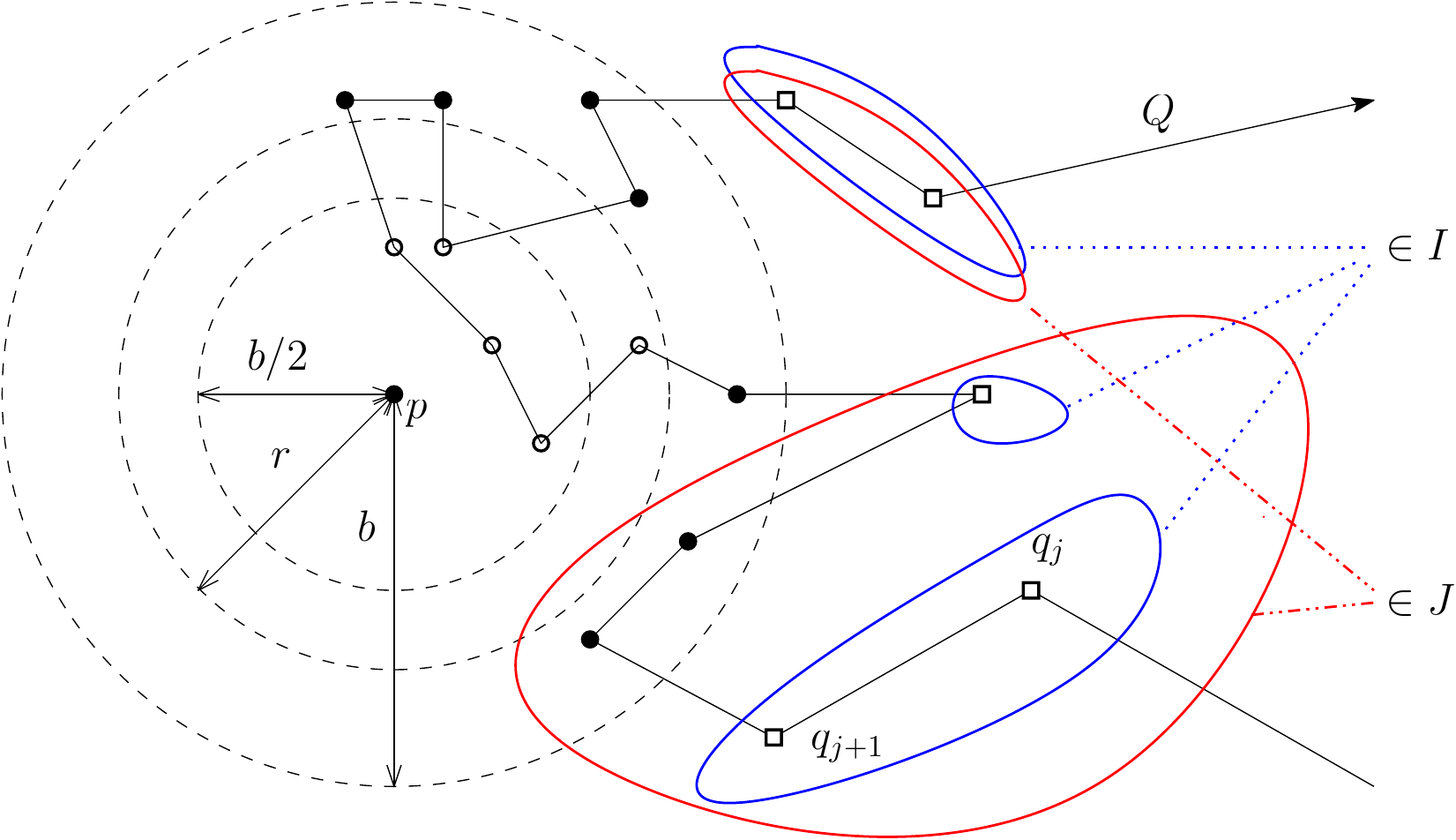}
\caption{The process of \lemref{comparationofareas} for the vertex $p$ and the curve $Q$}
\figlab{casesthetafig}
\end{figure}

\subsubsection{Avoidable pairs}

\begin{definition}[Avoidable pair]
\deflab{avoidablecell}
Let $B$ be the $\theta$-guarding set produced by Algorithm~\ref{guarding:theta}, and let $V=S_B\cup B\cup H_B$ be the partition of $V$ implied by $B$. The pair $(i,j)\in B$ is called avoidable if there exist a pair $(i',j')\in B$ and two partial traversals $T_1$ and $T_2$ of $P$ and $Q$ from $(1,1)$ to $(i',j')$, such that:
\begin{compactenum}[i)]
\item $\forall (i'',j'')\in \left( T_1\cup T_2\right) \setminus \lbrace (i',j') \rbrace$ it holds that $(i'', j'')\in S_B$,
\item there exist pairs $(i,y_1)\in T_1$ and $(i,y_2)\in T_2$, with $y_1<j<y_2$,
\item there exist pairs $(x_1,j)\in T_2$ and $(x_2,j)\in T_1$, with $x_1<i<x_2$.
\end{compactenum}
\end{definition}

We notice that for the pair to be avoidable, it suffices to have the conditions $i)$ and $ii)$, or $i)$ and $iii)$, since the remaining condition is implied by the monotonicity of the traversals. The definition of the avoidable pair $(i,j)$ implies that any partial traversal of $P$ and $Q$ from $(i,j)$ to $(t,t)$ has to have a nonempty intersection with $T_1\cup T_2$. 

%\begin{definition}[Obsolete pair]
%\deflab{obsoletecell}
%Let $B$ be a $\theta$-guarding set. The pair $(i,j)\in B$ is called obsolete iff $\lbrace (i,j+1), (i+1,j), (i+1,j+1) \rbrace \subseteq B\cup S_B$.
%\end{definition}

\figref{obsoleteavoidable} shows the pairs selected by Algorithm~\ref{guarding:theta} into the $\theta$-guarding set $B$, for some $\theta \geq 1$, marked with polygonal red and blue bounds. The pairs within the red bound are avoidable, and the pairs within the blue bound are not. Two partial traversals $T_1$ and $T_2$ in $S_B$ that make the red bounded pairs avoidable (as in \defref{avoidablecell}) are marked by arrows.
%\figref{obsoleteavoidable} shows that the notions of avoidable and obsolete pairs are not equivalent. The pairs with zeros form a $\theta$-guarding set $B$. The squared pairs are avoidable, but not obsolete, since the pair that has edges in $E$ to these three pairs is in $H_B$ and not in $B\cup S_B$. The circled pair is obsolete, but not avoidable, since there cannot exist two traversals in $S_B$ such that the properties of \defref{avoidablecell} are satisfied. 

\begin{figure}[ht]
\centering
%\[
%F_{\delta/\theta} = \left(\begin{matrix}
%\ldots& & & & & & & \\
%0&0&0&0&1&1&1&\ldots\\
%1&1&1&0&0&0&0&\ldots\\
%1&\fbox{0}&1&1&\circled{0}&1&0&\ldots\\
%1&\fbox{0}&\fbox{0}&1&1&1&0&\ldots\\
%1&1&1&1&1&1&0&\ldots
%\end{matrix}\right)
%\]
\[
F_{\delta/\theta} = \left(\begin{matrix}
\includegraphics[width=0.33\textwidth]{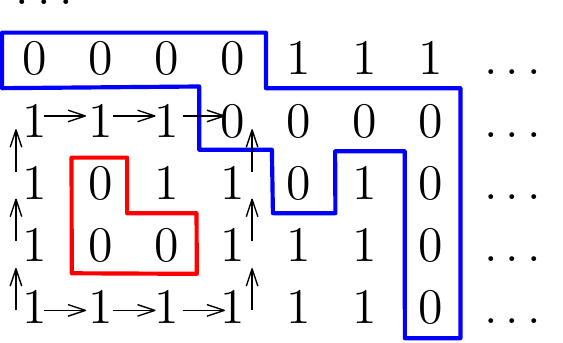}
\end{matrix}\right)
\]
\caption{Avoidable pairs from the $\theta$-guarding set $B$ (for some $\theta\geq 1$) are marked with red bound. Not avoidable pairs are marked with blue bound.
%A pair may be avoidable, but not obsolete; a pair may be obsolete, but not avoidable
}
\figlab{obsoleteavoidable}
\end{figure}

%\begin{lemma}
%\lemlab{trimminglemma}
%Given parameter $\theta\geq 1$ and the $\theta$-guarding set $B$. 
%Let $(i,j)\in B$ be an obsolete pair. Then $B\setminus \lbrace (i,j)\rbrace$ is a $\theta$-guarding set.
%\end{lemma}
%\begin{proof}
%The validity of the condition a) of \defref{guardingdiscrete} for the set
%$B\setminus \lbrace (i,j)\rbrace$ is inherited from the set $B$. In order to
%prove the condition b), for the sake of contradiction let there exist a
%traversal $T$ of $P$ and $Q$ such that $T\cap \left( B\setminus \lbrace
%(i,j)\rbrace\right) =\emptyset$. Since by \lemref{algguardingset} the
%traversal $T$ of $P$ and $Q$ satisfies $T\cap B\neq\emptyset$, it has to be
%$(i,j)\in T$. Since $(i,j)$ is obsolete, $(i+1,j)$, $(i,j+1)$ and $(i+1,j+1)$
%are all in $B\cup S_B$. If the successor of $(i,j)$ along $T$ (say $(i+1,j)$) is
%in $B$, then the traversal $B$ remains guarded by $B\setminus \lbrace
%(i,j)\rbrace$. If it would be in $S_B$, then it could be reached by a partial
%traversal $T'$ from $(1,1)$ using only pairs from $S_B$. Since $T\cap \left(
%B\setminus \lbrace (i,j)\rbrace\right) =\emptyset$, $T'$ could be extended from
%$(i+1,j)$ to $(n,n)$, such that $T'\cap B=\emptyset$, a contradiction. Thus,
%$B\setminus \lbrace (i,j)\rbrace$ is a $\theta$-guarding set.
%\end{proof}

\begin{lemma}
\lemlab{trimminglemma2}
Given parameter $\theta\geq 1$ and the $\theta$-guarding set $B$. 
Let $B'\subseteq B$ be the set of the avoidable pairs. 
Then $B\setminus B'$ is a $\theta$-guarding set.
\end{lemma}

\begin{proof}%[Proof of \lemref{trimminglemma2}]
The validity of the condition a) of \defref{guardingdiscrete} for the set $B\setminus B'$ is inherited from the set $B$. In order to prove the condition b), for the sake of contradiction let there exist a traversal $T$ of $P$ and $Q$ such that $T\cap \left( B\setminus B'\right) =\emptyset$. Since by \lemref{algguardingset} the traversal $T$ of $P$ and $Q$ satisfies $T\cap B\neq\emptyset$, there exists $(i,j)\in T\cap B'$, and we may assume that $(i,j)$ is the last such avoidable pair along $T$. Let $(i',j')\in B\setminus B'$, 
$T_1$ and $T_2$ be respectively the pair in $B$ and 
two traversals from \defref{avoidablecell} that make the pair $(i,j)$ avoidable. 

We may assume that $(i',j')$ is in $B\setminus B'$. To see this 
%it would be $(i',j')$ avoidable too. Let 
let $(i,j)=(i_1,j_1),(i_2,j_2),\ldots, (i_\ell, j_\ell)$ be the sequence of the pairs of indices, such that for all $m\in\lbrace 1,\ldots, \ell-1\rbrace$: 
\begin{compactenum}[a)]
\item the pair $(i_m,j_m)\in B'$; 
\item the pair $(i_{m+1},j_{m+1})$ makes the pair $(i_m,j_m)$ avoidable (from \defref{avoidablecell}); and
\item the pair $(i_\ell,j_\ell)\in B\setminus B'$.
\end{compactenum}
Such index $\ell$ has to exist, since it follows from \defref{avoidablecell} and from the monotonicity of traversals,  that $i_1<i_2<\ldots\leq n$ and $j_1<j_2<\ldots\leq n$. The partial traversals $T_1^{(\ell)}$ and $T_2^{(\ell)}$ from $(1,1)$ to $(i_\ell,j_\ell)$ given by \defref{avoidablecell}, that make the pair $(i_{\ell-1},j_{\ell-1})$ avoidable, satisfy the conditions of \defref{avoidablecell} for the pair $(i,j)$ as well. We assign $(i',j')=(i_\ell,j_\ell)\in B\setminus B'$, and thus it holds that $(i',j')\notin T$. 

Let $(\hat{i},\hat{j})\in T\cap (T_1\cup T_2)$ be the last such pair along $T$ (there has to exist at least one such pair, w.l.o.g let it be in $T_1$). We construct the traversal $T'$ of $P$ and $Q$ out of the partial traversal of $T_1$ from $(1,1)$ to $(\hat{i},\hat{j})$ and the partial traversal of $T$ from $(\hat{i},\hat{j})$ to $(t,t)$. For the pairs $(i'',j'')\in T'\cap T_1$ it holds by \defref{avoidablecell} that $(i'', j'')\in S_B$. Thus $(T'\cap T_1)\cap B=\emptyset$, since $B\cap S_B=\emptyset$.

But since $T\cap B=\emptyset$, it is also $(T'\cap T) \cap B =\emptyset$. Therefore for the traversal $T'$ it holds that $T'\cap B=\emptyset$. This contradicts the assumption that $B$ was the $\theta$-guarding set, and proves that the condition b) of \defref{guardingdiscrete} holds. Thus $B\setminus B'$ is a $\theta$-guarding set.
\end{proof}

%It is clear that if the set $B$ would be the set of pairs with zeros that are selected by Algorithm~\ref{guarding:theta} in the example of \figref{forklikecase}, we could not simply try to remove the obsolete pairs from the set $B$ to obtain the linear sized set in terms of the input size. Namely, none of the entries in such set $B$ is obsolete, since the entries in the third, sixth, etc. row are in $H_B$ and not in $S_B$ in the partition of $V$.

\subsubsection{Trimming the reachable area of a guarding set}

Let $B$ be a $1$-guarding set for two curves $P$ and $Q$.  We now want
to modify $B$ to shrink the number of pairs while maintaining the guarding
property. It turns out that we can do this if we relax the approximation
quality of the guarding set (which we denoted with $\theta$). We perform this trimming in three phases: 
\begin{compactenum}[(1)]
\item Remove all avoidable pairs from $B$. 
\item Trim the reachable area of $B$ row by row. 
\item Trim the reachable area of $B$ column by column.
%\item Iteratively remove all obsolete pairs from $B$. 
\end{compactenum}

In the following, we describe the trimming operation on a single row.  Consider a vertex
$p_i$ of the curve $P$ and consider the intersection of $B$ with the row of the
distance matrix associated with $p_i$. Let $I_i$ denote the set of intervals of
the column indices that represent this intersection. We now apply
\lemref{comparationofareas} with parameter $b=\distFr{P}{Q}$ to obtain
a set of intervals $J_i$ that can be used to trim the reachable area of $B$
with respect to the $i$th row. Each interval
in $J_i$ covers a set of intervals of $I_i$. Let $A_i$ be the subset of pairs
of the $i$th row of which the column index is 
contained in an interval of $J_i$, but not contained in any interval of $I_i$.
We call $A_i$ the \emph{filling pairs} of the row.
We now want to trim the reachable area $S_B$ defined by $B$ along the vertices of 
the reachability graph which correspond to pairs of $A_i$. For this we will 
remove all vertices of $B_i$ that are reachable from $A_i$ and add the 
pairs of $A_i$ to $B$.
See Algorithm~\ref{cutting:reach:row} for the pseudocode of this trimming operation. \figref{cuttingarea} illustrates the process with an example.
  The trimming operation for a single column is analogous, except that we use
  $b=\distFr{P}{Q}/2$ as a parameter to \lemref{comparationofareas}.

\begin{algorithm}[ht]
 \KwData{guarding set $B$, row index $i$, value of $b>0$}
 \KwResult{modified guarding set $B$}
 \caption{Trimming the reachable area for one row}\label{cutting:reach:row}
                $I_i := \{[j,j] ~|~ (i,j) \in B\}$\tcc*[r]{pairs of $B$ in the $i$th row}
                Let $J_i$ be the set of intervals obtained from \lemref{comparationofareas} using $I_i$ and $b=\distFr{P}{Q}$\;
                $A_i := S_B \cap \left\{(i,j)~|~ j \in \left(\bigcup_{[j_1,j_2] \in J_i} [j_1,j_2] \setminus \bigcup_{[i_1,i_2] \in I_i} [i_1,i_2] \right) \right\}$\tcc*[r]{Compute filling pairs}
 		FIFO-Queue $\mathcal{Q}\leftarrow A_i$\tcc*[r]{find guarding pairs reachable from $A_i$ via BFS}
		\While{$\mathcal{Q}\neq \emptyset$}{
			$(i,j)\leftarrow \text{pop}(\mathcal{Q})$\;
			\ForEach{$(i',j') \in \{ (i+1,j), (i+1,j+1)\}$}{			
				\If{$(i',j')\in B \setminus \mathcal{Q}$ }{$B \leftarrow B \setminus \lbrace (i',j')\rbrace$
                                \tcc*[r]{remove them from $B$}
				}
				\Else{
					push($\mathcal{Q},(i',j')$)\;
				}
			}
		}
        	$B \leftarrow B \cup A_i$\tcc*[r]{add pairs of $A_i$ to $B$}
\end{algorithm}

\begin{figure}[ht]
\centering
\[
F_{b}^{\text{before}} = \left(\begin{matrix}
\ldots& & & & & \\
0&0&\fbox{0}&\fbox{0}&\fbox{0}&\ldots\\
0&\fbox{0}&1&1&\fbox{0}&\ldots\\
0&1&1&\fbox{0}&0&\ldots\\
\fbox{0}&1&1&\fbox{0}&\fbox{0}&\ldots\\
1&1&1&1&1&\ldots
\end{matrix}\right)
\hspace*{1cm}
F_{b/2}^{\text{after}} = \left(\begin{matrix}
\ldots& & & & & \\
0&0&\circled{0}&\circled{0}&\circled{0}&\ldots\\
0&\circled{0}&1&1&\circled{0}&\ldots\\
0&1&1&\circled{0}&0&\ldots\\
\fbox{0}&\fbox{0}&\fbox{0}&\fbox{0}&\fbox{0}&\ldots\\
1&1&1&1&1&\ldots
\end{matrix}\right)
\]
\caption{The elements of a guarding set (marked with boxes) before (left) and after  (right) applying of Algorithm~\ref{cutting:reach:row} to the second row. The removed pairs are marked by circles}
\figlab{cuttingarea}
\end{figure}

\begin{lemma}
Let $B$ be a $1$-guarding set.
\begin{compactenum}[(i)]
\item After the first phase of the algorithm, which removes all avoidable pairs,
the modified set $B$ is a $1$-guarding set.  
\item After the second phase of the algorithm, which applies the trimming
operation to each row with $b=\distFr{P}{Q}$, the modified set $B$ is a
$2$-guarding set. 
\item After the third phase of the algorithm, which applies the trimming operation to each column with 
$b=\distFr{P}{Q}/2$, the modified set $B$ is a $4$-guarding set. 
%\item After the fourth phase of the algorithm, which removes all obsolete pairs, the modified set $B$ is a $4$-guarding set.  
\end{compactenum}
\lemlab{trimmingalgo}
\end{lemma}
\begin{proof}
The first part of the lemma follows directly from \lemref{trimminglemma2}.
We now prove the second part of the lemma statement. Condition (iii) of \lemref{comparationofareas} ensures
that any pair of a set $A_i$ added to $B$ corresponds to a pair of vertices $p\in P$ and $q\in Q$ with $d(p,q) \geq b/2 = \distFr{P}{Q}/2$. Indeed, the
column indices of the pairs of $A_i$ are contained in intervals of $J_i$.
Therefore, after the second phase, the modified set $B$ satisfies property (a)
in the definition of guarding sets if we set $\theta=2$.  Secondly, we argue
that property (b) is not invalidated after the trimming operation was applied
to a row. Let $B$ denote the guarding set before the trimming operation applied 
to the $i$th row and let $B'$ denote the modifed guarding set after trimming.
Clearly, the trimming operation does not add any avoidable pairs to $B$.
Therefore we can assume that throughout the second phase no 
avoidable pairs are present.

Assume for the sake of contradiction that there exists a traversal
$T$ that contains a pair of $B$, but does \emph{not} contain a pair of $B'$. 
Let $(i',j')$ be the first pair along $T$ that was
removed from $B$ during the trimming operation and let $(i,j_2)$ be a pair of 
$A_i$ that has a BFS-path to $(i',j')$. 
$T$ must contain a pair $(i,j_1)$ in the $i$th row and this pair cannot be contained in
an interval of $J_i$ (otherwise $T$ would contain a pair of $B'$).
Let $T_1$ be the partial traversal (path in $G$) of $T$ that starts in $(1,1)$
goes via $(i,j_1)$ and ends in $(i',j')$. Since $(i',j')$ was the first vertex
along $T$ in $B$, it follows that $T_1$ only visits vertices that are in $S_B$.
Note that $i'>i$ since the BFS only visits row indices strictly greater $i$.
Since $A_i \subseteq S_B$, there must be a path $T_2$ in $G$ 
from $(1,1)$ via $(i,j_2)$ to $(i',j')$ that only contains vertices of $S_B$.
Now, condition (ii) of \lemref{comparationofareas} implies that there must be
a vertex $(i,j'')$ in $B$, such that either $j_1 < j''< j_2$ or
$j_2 < j''< j_1$. This implies that $(i,j'')$ must be avoidable with respect to $B$.
However, this contradicts the fact that $B$ does not contain any avoidable pairs.
This proves (ii).  The third part of the lemma
follows by a symmetric argument applied to the columns.
%The last part of the lemma follows directly from \lemref{trimminglemma}.
\end{proof}

%here was the section 3.2.5

\subsection{Bounding the complexity of the modified guarding set}
\seclab{boundingcomplexity}

Given set $B$ after the algorithm of \lemref{trimmingalgo}. For every row of $B$ (presented as matrix) let the pairwise disjoint set of intervals $R_i\subseteq \lbrace \left[ j_1,j_2\right] | j_1\leq j_2\in \mathbb{N}, 1\leq j_1\leq j_2\leq t\rbrace$ be a set of intervals on $\lbrace 1,\ldots, t\rbrace$ of minimal size, such that for any $1\leq j' \leq t$ there exist $j_1$ and  $j_2$ with  $j' \in [j_1,j_2] \in R_i$  if and only if $(i,j') \in B$. We can analogously define such pairwise disjoint sets $C_j$ over the columns of $B$. 

\lemref{comparationofareas} implies that for every row $i$ there is a set of pairwise disjoint intervals $J_i$ constructed by line 2 of Algorithm~\ref{cutting:reach:row}, with $|J|\leq c+1$. Algorithm~\ref{cutting:reach:row} takes into $B$ only the pairs that belong to the subsets of the intervals of $J_i$ that were in $S_B$ too. But since the pairs $(i,j)\in H_B$ such that $j\in [j_1,j_2]\in J_i$ have the property that any traversal using these pairs has to contain a pair in $B$ prior to $(i,j)$, we could have added such pairs too into $B$ and then it would be $J_i=R_i$. Since we took only its subsets, it holds that for every $[j_1,j_2]\in R_i$ there is $[j_3,j_4]\in J_i$ with $j_3\leq j_1\leq j_2\leq j_4$. By counting all intervals of $R_i$ that are subset of one interval from $J_i$ as one, we say that all such intervals $R_i$ build one \emph{extended group} of consecutive pairs within $i$th row. It follows that there are at most $c+1$ extended groups within $i$-th row. This process gets repeated over columns as well. See \figref{extendedgroups} for an illustration.

\begin{figure}[ht]
\centering
\includegraphics[width=0.35\textwidth]{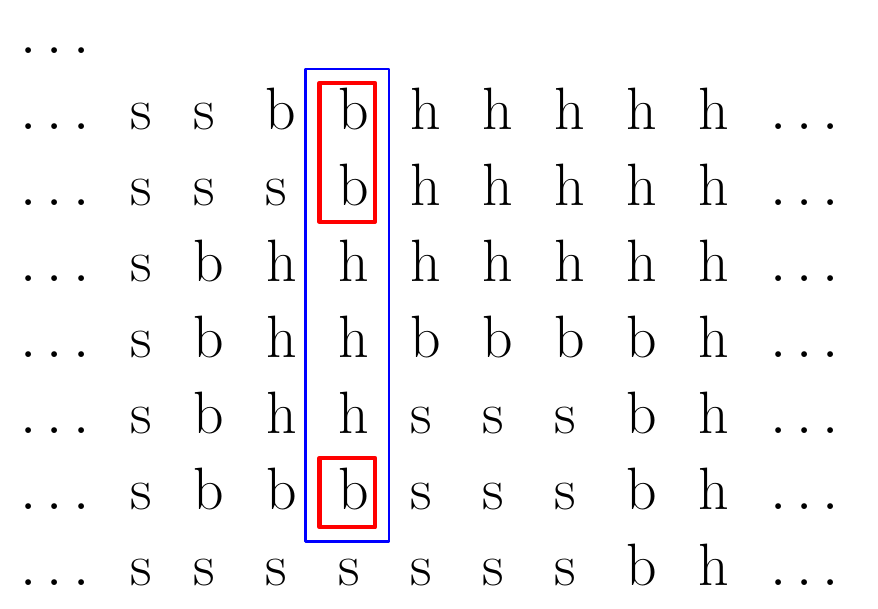}
\caption{The pairs of the guarding set $B$ (red) and its extended group (blue) within one column. The pairs denoted with s, b, and h are from $S_B$, $B$ and $H_B$ respectively}
\figlab{extendedgroups}
\end{figure}

We have to note that the filling pairs added into $B$ also imply the removal of a pair in $B$ that lies in the same row but with higher column index, except possibly for the last pair in the row. This can happen at most once per row, adding one pair (and one extended group) to the row. We obtain the following lemma.

\begin{lemma}\lemlab{extendedgroups}
In the guarding set produced by Algorithm~\ref{guarding:theta} and modifed by
the algorithm of \lemref{trimmingalgo}, there are at most $c+1$ extended groups
within a column, and $c+2$ extended groups within a row.
\end{lemma}

To finally bound the complexity of our guarding set by \lemref{extractobsolete}, we show first \lemref{predecessorins}.

\begin{lemma}\lemlab{predecessorins}
For the guarding set produced by Algorithm~\ref{guarding:theta} and after every phase of algorithm of \lemref{trimmingalgo} the following invariant holds: for every pair $(i,j)\in B$ there exists a pair $(i',j')\in S_B$ such that $((i',j'),(i,j))\in E$.
\end{lemma}
\begin{proof}[Proof of \lemref{predecessorins}]
We call the pair $(i',j')$ the predecessor pair. The construction of the guarding set $B$ Algorithm~\ref{guarding:theta} guarantees that a pair $(i,j)$ is added into $B$ if it is visited over an edge $((i',j'),(i,j))\in E$, where $(i',j')\notin B$. Thus $((i',j')\in S_B$ as claimed. 

The first phase of the algorithm of \lemref{trimmingalgo} removes the avoidable pairs from $B$, thus for the pairs that remain in $B$ the invariant holds. The second phase runs Algorithm~\ref{cutting:reach:row} upon a row and adds into $B$ only pairs which were already in $S_B$, thus have also a predecessor in $S_B$. For every pair $(i',j')$ which was in $S_B$ before and is in $H_B$ after Algorithm~\ref{cutting:reach:row} it holds that the BFS passes it and then visits and subsequently removes the pairs from $B$. Therefore the invariant remains valid for the pairs that remain in $B$, as for the pairs that were already in $B$ their predecessors remain in $S_B$, so their status is not changed. The third phase is equivalent to the second one, 
%The fourth phase removes obsolete pairs, thus it assigns the pairs from $B$ to $S_B$, 
and the invariant remains valid.
\end{proof}

\begin{lemma}
\lemlab{extractobsolete}
The set $B$ obtained by the algorithm of \lemref{trimmingalgo} is a $4$-guarding set, containing at most $(3c+4)\cdot t$ pairs.
\end{lemma}

\begin{proof}%[Proof of \lemref{extractobsolete}]
For every pair $(i,j)\in B$ one of the following holds true: 
\begin{compactenum}[i)]
\item the index $j$ is the smallest index of an extended group over the $i$th row;
\item the index $i$ is the smallest index of an extended group over the $j$th column; 
\item none of the above.
\end{compactenum}

We argue that if neither i) nor ii) holds true, then it must be that $i-1$ is
the smallest index of an extended group over the $j$th column.  Indeed, note
that if neither i) nor ii) holds true, then $(i-1,j)$ and $(i,j-1)$ are part of
an extended group and such groups can only contain pairs of $B$ or $H_B$.
Therefore, the pair $(i-1,j-1)$ must be in $S_B$ because
\lemref{predecessorins} implies that $(i,j)$ must have an ingoing edge from a
pair in $S_B$.  Now, since pairs of $S_B$ and $H_B$ cannot be directly
connected by an edge of $G$, it must be that $(i-1,j)$ and $(i,j-1)$ are both
in $B$. Thus, $i-1$ is the smallest index of an extended group over the $j$th
column.

We charge elements of $B$ of type i) and of type ii) to their respective
extended intervals. We charge elements of type iii) it to their extended interval 
over the column. Thus, extended intervals in the column are charged at most twice.
By \lemref{extendedgroups} we have at most $(c+1)$ extended intervals per column and at most
$(c+2)$ extended intervals per row. This implies that altogether $|B|\leq
(3c+4)\cdot t$, as claimed.
\end{proof}

%Using counting argument we prove \lemref{extractobsolete}, which claims the linear upper bound on the complexity of our modified guarding set $B$. The complete proof can be found in \subsecref{boundingcomplexity} in the appendix.
\lemref{k:subset:approx} and \lemref{extractobsolete} imply the correctness of \thmref{claim:discfrec:n:approx}.
The proof of \thmref{claim:discfrec:n:approx45} is analogous to the proof of \thmref{claim:discfrec:n:approx}, while \lemref{varphibound} and \lemref{k:subset:approx} are replaced by \lemref{varphibound45} and \lemref{k:subset:approx45}, respectively. The rest of the proof can be taken verbatim.

\section{Lower bounds}
\seclab{lowerbounds}

\subsection{Definitions}

Related similarity measure between two curves to the discrete \Frechet distance is dynamic time warping. It considers the \emph{sum} of the used distances in the traversal (instead the maximum one). Formally, for two curves $P$ and $Q$ from $\mathbb{R}^d$, we define:
\begin{equation}
\distDTW{P}{Q}=\min_{ T \in \mathcal{T}} \sum_{(i,j)\in T} \|p_i-q_j \|.
\end{equation}

For the continuous \Frechet distance, let again $P=\lbrace p_1,\ldots, p_t\rbrace$ and $Q=\lbrace q_1,\ldots, q_t\rbrace$ be two curves from $\mathbb{R}^d$. Let $\pi:[0,1]\rightarrow P$ and $\tau:[0,1]\rightarrow Q$ be two functions on $[0,1]$ such that $\pi(0)=p_1$, $\pi(1)=p_t$, $\tau(0)=q_1$ and $\tau(1)=q_t$, and such that $\pi$ and $\tau$ are monotone on $P$ and $Q$ respectively. Let $\mathcal{H}$ denote the set of continuous and increasing functions
$f:[0,1]\rightarrow[0,1]$ with the property that $f(0)=0$ and $f(1)=1$. 
For two given curves $P$ and $Q$ and respective functions
$\pi$ and $\tau$,
their (continuous) \emph{\Frechet distance} is defined as 
\begin{equation} \label{def:frechet} 
\distFr{P}{Q}=\inf_{f\in \mathcal H}\; \max_{t \in [0,1]} \| Q(\tau(f(t)))-P(\pi(t))
\|
\end{equation}
and the function $f$ that reaches the value $\delta=\distFr{P}{Q}$ is called \emph{matching} from $P$ to $Q$ with cost $\delta$.

%It may happen that for some two curves $P$ and $Q$ it holds that the ratio between \Frechet distance of the curves and the \Frechet distance of the respective projection curves $P'$ and $Q'$ is at least $\Omega(n)$, where $n$ is the complexity of the curves. This claim holds for both the discrete and the continuous version of the \Frechet distance. The similar claim can be made for the dynamic time warping distance too.

\subsection{c-packed curves}

We prove the correctness of \thmref{lowerbounddiscretecpacked} for the discrete and the continuous \Frechet distance, as well as for the dynamic time warping distance.

\begin{proof}[Proof of \thmref{lowerbounddiscretecpacked} for the discrete \Frechet distance]
Let the curves $P$ and $Q$ be from $\mathbb{R}^2$. Let the curve $P=\lbrace p_1,\ldots, p_{2t+1}\rbrace$ be the line segment $\overline{p_1p_{2t+1}}$, while the vertices $p_2,\ldots, p_{2t}$ are uniformly distributed on $P$, i.e. $\|p_{i+1}-p_i\|=\|p_i-p_{i-1}\|$ for all $i\in\lbrace 2,\ldots, 2t\rbrace$. Let $Q=\lbrace q_1,\ldots, q_{2t+1}\rbrace$ be composed by two line segments $\overline{q_1q_{t+1}}$ and $\overline{q_{t+1}q_{2t+1}}$, and the vertices $q_2,\ldots, q_{2t}$ are uniformly distributed on $Q$, i.e. $\|q_{j+1}-q_j\|=\|q_j-q_{j-1}\|$ for all $j\in\lbrace 2,\ldots, 2t\rbrace$. Let $p_1=q_1$ and $p_{2t+1}=q_{2t+1}$ and let $\angle q_{t+1}q_1p_{2t+1}=\alpha$ (as shown in \figref{lowerbounddiscrete}).

\begin{figure}[ht]
\centering
\includegraphics[width=0.75\textwidth]{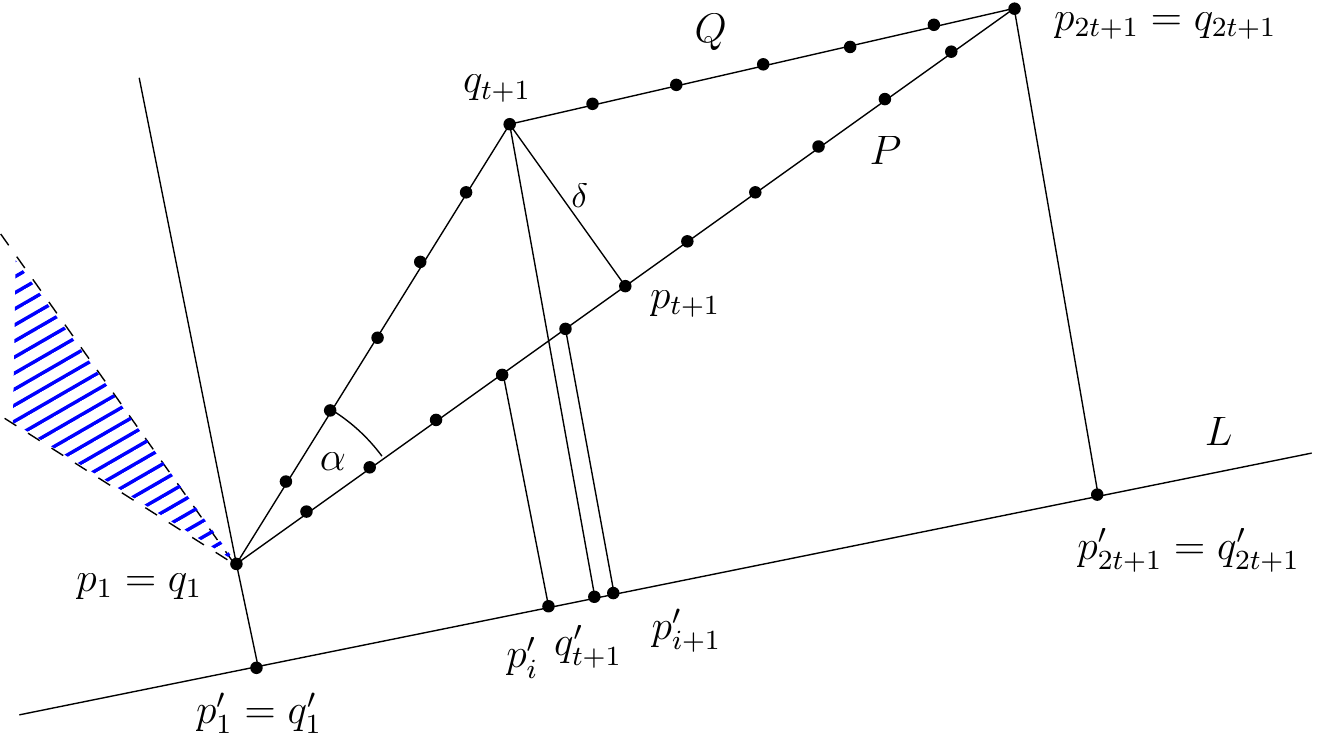}
\caption{Lower bound for the discrete \Frechet distance case for $c$-packed curves}
\figlab{lowerbounddiscrete}
\end{figure}

The curves $P$ and $Q$ are $c$-packed for any constant $c\geq 2$.  Let $|Q|=2$, then it holds that $|P|=2\cdot|\cos\alpha|$ and for both discrete and continuous \Frechet distance it holds that $\delta=\distFr{P}{Q}=|\sin\alpha|$.

Let the straight line $L$ support the unit vector $\textbf{u}$, which is chosen uniformly at random on the unit hypersphere, and let $P$ and $Q$ be projected to $L$. 
Observe that the discrete \Frechet distance of $P$ and $Q$ is realized by the pair $(t+1, t+1)$ in the traversal of $P$ and $Q$, thus $\| p_{t+1}-q_{t+1}\| = \distFr{P}{Q}=\delta$. The vertex $q_{t+1}\in Q$ is projected to $q'_{t+1}$ and $q'_{t+1}$ lies either within $P'$ or outside of it.

If it is inside ($q'_{t+1}\in P'$), and thus in one of the $2t$ line segments $\overline{p'_i p'_{i+1}}$ for some $i\in \lbrace 1,\ldots, 2t\rbrace$, then the distance of $q'_{t+1}$ to its matched vertex $p'_x\in P'$ is at most
\[
\|q'_{t+1}-p'_x\| \leq \frac{|P'|}{2t}\leq \frac{|P|}{2t}=\frac{ |\cos\alpha|}{t}.
\]
Therefore it holds that
\[
q'_{t+1}\in P' \Rightarrow \distFr{P'}{Q'}\leq \frac{|\cos\alpha|}{t} \leq \frac{1}{t}.
\]
The event $q'_{t+1}\in P'$ occurs with probability at least $1-\alpha/\pi$, i.e. when the perpendicular line to $L$ is not parallel to some straight line laying in $\angle q_{t+1}q_1p_{2t+1}=\alpha$ and including $q_1$ (tiled area in \figref{lowerbounddiscrete}). Then it holds that 
\[
\text{Pr}\left[ \frac{\distFr{P}{Q}}{\distFr{P'}{Q'}} \geq |\sin\alpha|\cdot t\right] \geq 1- \frac{\alpha}{\pi}.
\]
For $\alpha\in [0,1]$ it holds that $|\sin \alpha| \geq \alpha-\alpha^3/3!\geq \frac{5}{6}\alpha$, thus for $\gamma=\alpha/\pi$ is for $\gamma\in (0,1/\pi)$:
\[
\text{Pr}\left[ \frac{\distFr{P}{Q}}{\distFr{P'}{Q'}} \geq \frac{5\pi\gamma}{6}\cdot t\right] \geq 1- \gamma
\]
This proves the correctness of the theorem.

\end{proof}

\begin{proof}[Proof of \thmref{lowerbounddiscretecpacked} for the continuous \Frechet distance]
For the continuous case it holds that if $q'_{t+1}\in P'$, then $P'=Q'$ and $\distFr{P'}{Q'}=0$. Thus it holds that 
\[
\text{Pr}\left[ \frac{\distFr{P}{Q}}{\distFr{P'}{Q'}} \geq t\right] \geq \text{Pr}\left[ \distFr{P'}{Q'} =0\right] \geq 1-\alpha/\pi
\]
for any constant $\alpha\in (0,1)$. Thus the continuous \Frechet distance will be reduced at least by a factor of $t$ with  probability at least $1-\gamma$, where $\gamma=\alpha/\pi$ and $\gamma\in (0,1/\pi)$.
\end{proof}

\begin{proof}[Proof of \thmref{lowerbounddiscretecpacked} for the dynamic time warping distance]
For the curves $P$ and $Q$ it holds that
\begin{eqnarray*}
\distDTW{P}{Q}&=&\sum_{i=1}^{2t+1} \| p_i - q_i\| = 2\cdot \left(\sum_{i=2}^{t} \| p_i - q_i\|\right) + \| p_{t+1} - q_{t+1}\| \\
&=& 2\cdot \left(\sum_{i=1}^{t} \| p_{i+1} - q_{i+1}\| \right) - \| p_{t+1} - q_{t+1}\| = 2\cdot \left(\sum_{i=1}^{t} \frac{i\cdot |\sin \alpha|}{t}\right) - |\sin\alpha|  \\
&=& t\cdot |\sin\alpha|
\end{eqnarray*}
For the projection curves it holds that with the probability $1-\alpha/\pi$ that (analogously to the discrete \Frechet distance case):
\begin{eqnarray*}
\distDTW{P'}{Q'}&=&\min_{ T \in \mathcal{T'}} \sum_{(i,j)\in T} \|p'_i-q'_j \| %\leq  \sum_{i=1}^{2n+1} \|p'_i-q'_i \|
\end{eqnarray*}
where $\mathcal{T'}$ is the set of all traversals of $P'$ and $Q'$.

Let the set of the pairs $T'$ be defined, such that for $1\leq j\leq 2t+1$, the pair $(i,j)\in T'$ iff $\|p'_i-q'_j\|$ is minimal over all $1\leq i\leq 2t+1$. Such set $T'$ is a traversal of $P'$ and $Q'$. This is shown by induction, since $p_1=q_1$ and $p_{2t+1}=q_{2t+1}$. Let the pair $(i,j)$ be in $T'$. Then the closest vertex of $P'$ to the vertex $q'_{j+1}$ has to be either $p'_i$ or $p'_{i+1}$. The other ones (either with smaller or greater index) cannot be the closest ones to $q'_{j+1}$ because of the order of the vertices on $P'$ and $Q'$. Thus the pair $(i,j)\in T'$ is followed either by $(i+1,j+1)$ or $(i,j+1)$ (the possibility of $(i+1,j)$ is excluded, since we choose exactly one matched vertex for each $j$, $1\leq j\leq 2t+1$), and $T'\in \mathcal{T}'$ is a traversal.

Therefore it holds that 
\begin{eqnarray*}
\distDTW{P'}{Q'}&\leq & \sum_{(i,j)\in T'} \|p'_i-q'_j \| \leq \frac{1}{2}\sum_{(i,j)\in T'} \|p'_i-p'_{i+1} \|\\
 &\leq & \frac{1}{2}\sum_{i=2}^{2t} \|p'_i-p'_{i+1} \| \leq \frac{1}{2}\cdot |P'| \leq  \frac{1}{2}\cdot |P| = |\cos \alpha|\leq 1
\end{eqnarray*}
with the probability $1-\alpha/\pi$. Thus
\[
\text{Pr}\left[ \frac{\distDTW{P}{Q}}{\distDTW{P'}{Q'}} \geq |\sin\alpha|\cdot t\right] \geq 1- \frac{\alpha}{\pi}.
\]
By repeating the analysis of \thmref{lowerbounddiscretecpacked} for the discrete \Frechet distance we obtain that the dynamic time warping distance will be reduced at least by a factor of $5\pi\gamma t /6$ with probability at least $1-\gamma$, for any $\gamma\in (0,1/\pi)$.
\end{proof}

\subsection{General case curves}

If the curves $P$ and $Q$ are not $c$-packed, for any constant $c\geq 2$, then the ratio of the continuous \Frechet distances between $P$ and $Q$ and their projection curves $P'$ and $Q'$ can be at least linear in $t$, as claimed by \thmref{lowerboundcontinuousgeneral}. This event can happen with probability 1. We claim the same bound for the discrete \Frechet distance. 

\begin{theorem}
\thmlab{lowerboundcontinuousgeneral}
There exist the curves $P=\lbrace p_1,\ldots, p_t \rbrace$ and $Q=\lbrace q_1,\ldots, q_t\rbrace$, such that if $P'$ and $Q'$ respectively are their projections to the one-dimensional space that supports the unit vector chosen uniformly at random on the unit hypersphere, then it holds that 
\[
\frac{\distFr{P}{Q}}{\distFr{P'}{Q'}}\geq f(t),
\]
where $f(t)\in \Omega(t)$.
\end{theorem}

\begin{proof}[Proof of \thmref{lowerboundcontinuousgeneral} for the continuous \Frechet distance]
We denote with $P_{k}$ the star-like clo\-sed curve with $2k+1$ vertices, defined as $P_{k}=\lbrace p_1,p_0,p_2, $ $p_0,\ldots,p_k,p_0,p_{k+1}\rbrace$. Let $p_i=(r_i,\theta_i)$ in polar coordinates be defined as $p_0=(0,0)$, $p_{k+1}=p_0$ and $p_i=(1,2\cdot (i-1)\cdot\pi/k)$ for $1\leq i\leq k$. Let $P=P_{k}$ and $Q=P_{k+1}$, and let $k$ be even. To have the same complexity for $P$ and $Q$ we can add two more points $p_1$ at the end, thus $t=2k+3$. We denote the indices of the curve $Q$ with $q_j$, $0\leq j\leq k+2$. \figref{example:buchin2} shows the curves $P$ and $Q$ for $k=12$ (in full blue and dotted red line respectively). 

\begin{figure}[ht]
\centering
\includegraphics[width=0.55\textwidth]{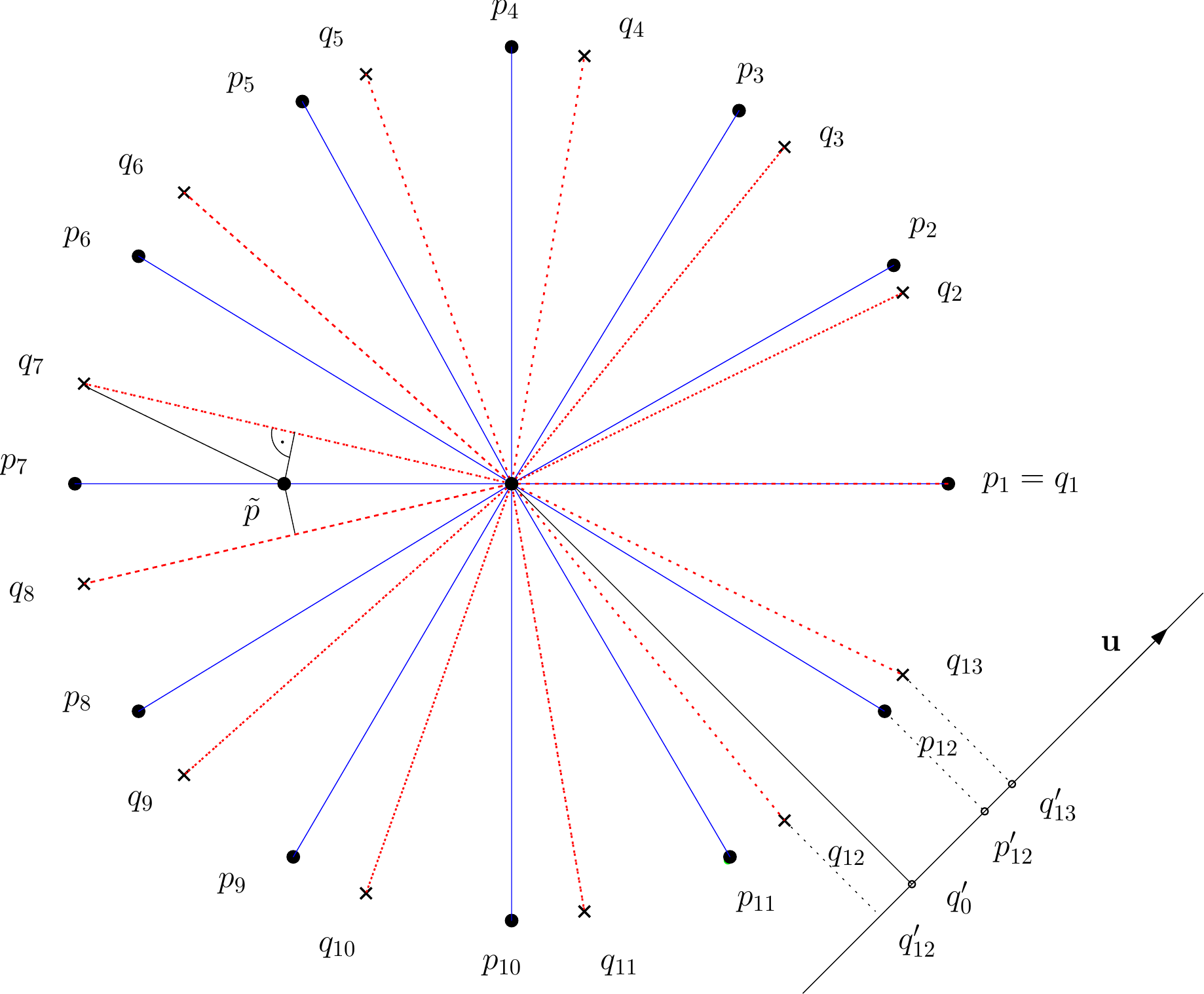}
\caption{Two curves $P$ and $Q$ with parameter $k=12$}
\figlab{example:buchin2}
\end{figure}

The \Frechet distance between the curves $P$ and $Q$ is $\distFr{P}{Q}=1/\left(2\cdot\cos (\pi/(k+1))\right)$. To show this, let $M$ be the matching of the points of $P$ and $Q$ that realizes the \Frechet distance. The curve $Q$ has one more ``ray'' of the star to be traversed. The ``rays'' $\overline{ p_0, p_1, p_0}$ and $\overline{ q_0, q_1, q_0}$ are equal, they are matched by $M$ at distance 0; and the ``rays'' $\overline{ p_0, p_i, p_0}$ and $\overline{ q_0, q_i, q_0}$ for $1\leq i\leq k/2$, and $\overline{ p_0, p_i, p_0}$ and $\overline{ q_0, q_{i+1}, q_0}$ for $1\leq i\leq k/2+2$ are pairwise matched by $M$ at distance smaller than $\pi/(k+1)$. There remain two consecutive ``rays'' $\overline{q_0q_jq_0}$ and $\overline{q_0q_{j+1}q_0}$ that have to be matched by the matching $M$ to $\overline{p_0p_jp_0}$, with $j=k/2+1$. The point $\tilde{p}$ with coordinates $\left(1/\left(2\cdot\cos (\pi/(k+1))\right), \pi \right)$ is the intersection of a bisector of $\overline{q_0 q_j}$ with $\overline{p_0 p_j}$. Such point $\tilde{p}$ matches the subcurve of $Q$ between the vertices $\lbrace q_{k/2+1}, q_0, q_{k/2+2}\rbrace$, thus the matching $M$ is completely described, and the \Frechet distance realized by $M$ is $\| q_{k/2+1} - \tilde{p}\|=\| p_0 - \tilde{p}\|=1/\left(2\cdot\cos (\pi/(k+1))\right)$, as claimed. It holds that $\distFr{P}{Q}> 1/2$ for any $k\geq 2$.

We notice that between every two lines $\overline{q_0q_j}$ and $\overline{q_0q_{j+1}}$ there has to be one line $\overline{p_0p_i}$ (the opposite does not have to hold). Thus the distance between $p_i$ and any of its neighboring $q_j$ and $q_{j+1}$ is at most $\max\lbrace \|q_{j}-p_{i}\|, \|q_{j+1}-p_{i}\|\rbrace  \leq \| q_{j+1}-q_{j}\| \leq 2\pi/(k+1)$, since $p_i$ is on the circular arc between $q_j$ and $q_{j+1}$.

If we now project the curves $P$ and $Q$ to the straight line that supports the unit vector $\textbf{u}$, with $\textbf{u}$ chosen uniformly at random on the unit hypersphere, let $P'=\lbrace p'_1,p'_0,\ldots,p'_k,p'_0,p'_{k+1}\rbrace$ and $Q=\lbrace q'_1,q'_0,\ldots,q'_{k+1},q'_0,q'_{k+2} \rbrace$ be their projections respectively. The line $\overline{q_0q'_0}$ satisfies one of the following two cases:
\begin{compactenum}[i)]
\item $\overline{q_0q'_0} \| \overline{q_0q_j}$ for some $1\leq j\leq k+1$, or
\item $\overline{q_0q'_0}$ lies between $\overline{q_0q_j}$ and $\overline{q_0q_{j+1}}$ for some $1\leq j\leq k+1$.
\end{compactenum}
Then in the first case, since $k$ is even, the straight line $\overline{q_0q'_0}$ lies between $\overline{q_0q_{(j+k/2)\mod (k+1)}}$ and $\overline{q_0q_{1+(j+k/2)\mod (k+1)}}$ (through the two vertices on the opposite side of the star). Therefore, we may only consider the second case.

The projected curves $P'$ and $Q'$ can be matched by a matching $M'$ as follows: let $p_i$ be the vertex of $P$ that lies between $\overline{q_0q_j}$ and $\overline{q_0q_{j+1}}$ from the case definition. Let $p'_i$ be its projection. Then let the subcurves $\lbrace p'_0,p'_i,p'_0\rbrace$ and $\lbrace q'_0,q'_{j},q'_0,q'_{j+1},q'_0\rbrace$ be matched to each other by matching $M'$. For the rest of the curves let $\lbrace p'_0,p'_{i+ \ell},p'_0\rbrace$ and $\lbrace p'_0,p'_{i- \ell},p'_0\rbrace$ be matched to $\lbrace q'_0,q'_{j+1+\ell},q'_0\rbrace$ and $\lbrace q'_0,q'_{j-\ell},q'_0\rbrace$ respectively, where $1\leq i-\ell$ or $i+\ell\leq k+1$. 

Let $M'(p'_i)$ be the point on $Q'$ that is matched to $p'_i$. Let $\overline{M}(p'_i)$ be the point on $Q$ such that $M'(p'_i)$ is its projection on $L$. If we denote with $\alpha_i$ the angle between the vector $\overline{M}(p'_i)-p_i$ and the unit vector $\textbf{u}$, then for the \Frechet distance between the projections $P'$ and $Q'$ (that lay in the one-dimensional space) it holds that
\begin{eqnarray*}
\distFr{P'}{Q'}&\leq& \max_{1\leq i\leq k}\lbrace \| M'(p'_i)-p'_i\|\rbrace = \max_{1\leq i\leq k}\lbrace \| \overline{M}(p'_i)-p_i\|\cdot |\cos \alpha_i|\rbrace \\
&\leq& \max_{1\leq i\leq k}\lbrace \| \overline{M}(p_i)-p_i\|\rbrace
\leq \frac{2\pi}{k}
\end{eqnarray*}
Therefore by projecting the curves $P$ and $Q$ to any straight line the \Frechet distance between the curves will be diminished at least by the factor
\[
\frac{\distFr{P'}{Q'}}{\distFr{P}{Q}}< \frac{2\pi}{k}\cdot 2=\frac{4\pi}{k}.
\]
This yields the claimed linear lower bound, since $k=(t-3)/2$  and proves the theorem with $f(t)=(t-3)/(8\pi)$.
\end{proof}

\begin{proof}[Proof of \thmref{lowerboundcontinuousgeneral} for the discrete \Frechet distance]
The lower bound given by \thmref{lowerboundcontinuousgeneral} holds also for the discrete \Frechet distance, with $f(t)=(t-5)/(16\pi)$. We adapt the curves $P$ and $Q$ from the proof for the continuous \Frechet distance as follows. Let us add to each ``ray'' $\overline{p_0 p_i p_0}$ of the curve $P_k$ the vertices $\hat{p}_i$ (i.e. the ``ray'' becomes $\overline{p_0 \hat{p}_i p_i \hat{p}_i p_0}$), with polar coordinates $\hat{p}_i=\left( 1/\left(2\cdot\cos (\pi/(k+1)\right),\right.$ $\left. 2\cdot (i-1)\cdot \pi/k\right)$. The curve $P_k$ contains now $4k+1$ vertices and $t=4k+5$. The rest of the construction and analysis can be used verbatim.
\end{proof}

\section{Experiments}
\seclab{experimentalpart}

We performed the preliminary experiments on the dataset of the 6th ACM SIGSPATIAL GISCUP 2017 competition\footnote{
\url{http://sigspatial2017.sigspatial.org/giscup2017/download}, downloaded on February 7th, 2018}. Their dataset $\mathcal{D}$ contains 20199 realistic polygonal curves from $\mathbb{R}^2$, with complexities between 9 and 767. We have repeated the following procedure for 504 pairs of curves of $\mathcal{D}$ selected uniformly at random. For each pair of curves (or their subcurves) the projection line was sampled $r=1000$ times. We observed the obtained distribution of the distortion $c$ of the discrete \Frechet distance.
\begin{compactenum}[i)]
\item We calculated the distortion $c=\distFr{P'}{Q'}/\distFr{P}{Q}$ for the whole curves.
\item We observed the prefix curves $P_\ell = \lbrace p_1,\ldots,p_\ell\rbrace$ and $Q_\ell = \lbrace q_1,\ldots,q_\ell\rbrace$ of $P$ and $Q$ respectively, with complexity $\ell$ equal 10, or to the multiples of $50$. The distortion $c=\distFr{P'_\ell}{Q'_\ell}/\distFr{P_\ell}{Q_\ell}$ is calculated.
\item For every prefix length $\ell$ we chose at random subcurves of $P$ and $Q$ of complexity $\ell$, defined by $\ell$ consecutive vertices of $P$ and $Q$ respectively. Let these curves be $P_{\ell,r}$ and $Q_{\ell,r}$. We calculated the distortion $c=\distFr{P'_{\ell,r}}{Q'_{\ell,r}}/\distFr{P_{\ell,r}}{Q_{\ell,r}}$.
\end{compactenum} 
This yielded 4286 pairs of (sub)curves. 

\begin{figure}[ht]
\centering
\includegraphics[width=0.49\textwidth, height= 0.36\textwidth]{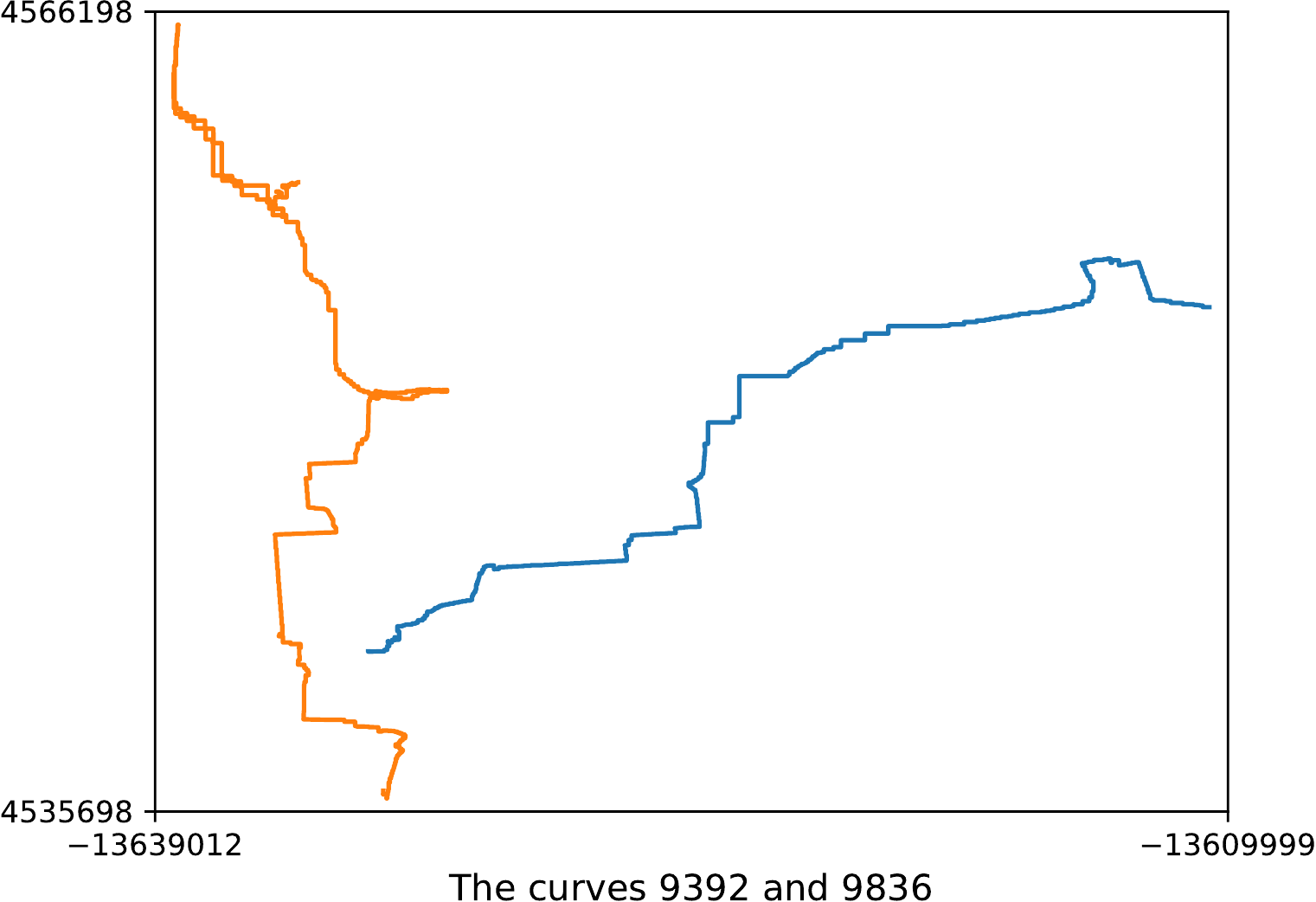}
\includegraphics[width=0.49\textwidth, height= 0.36\textwidth]{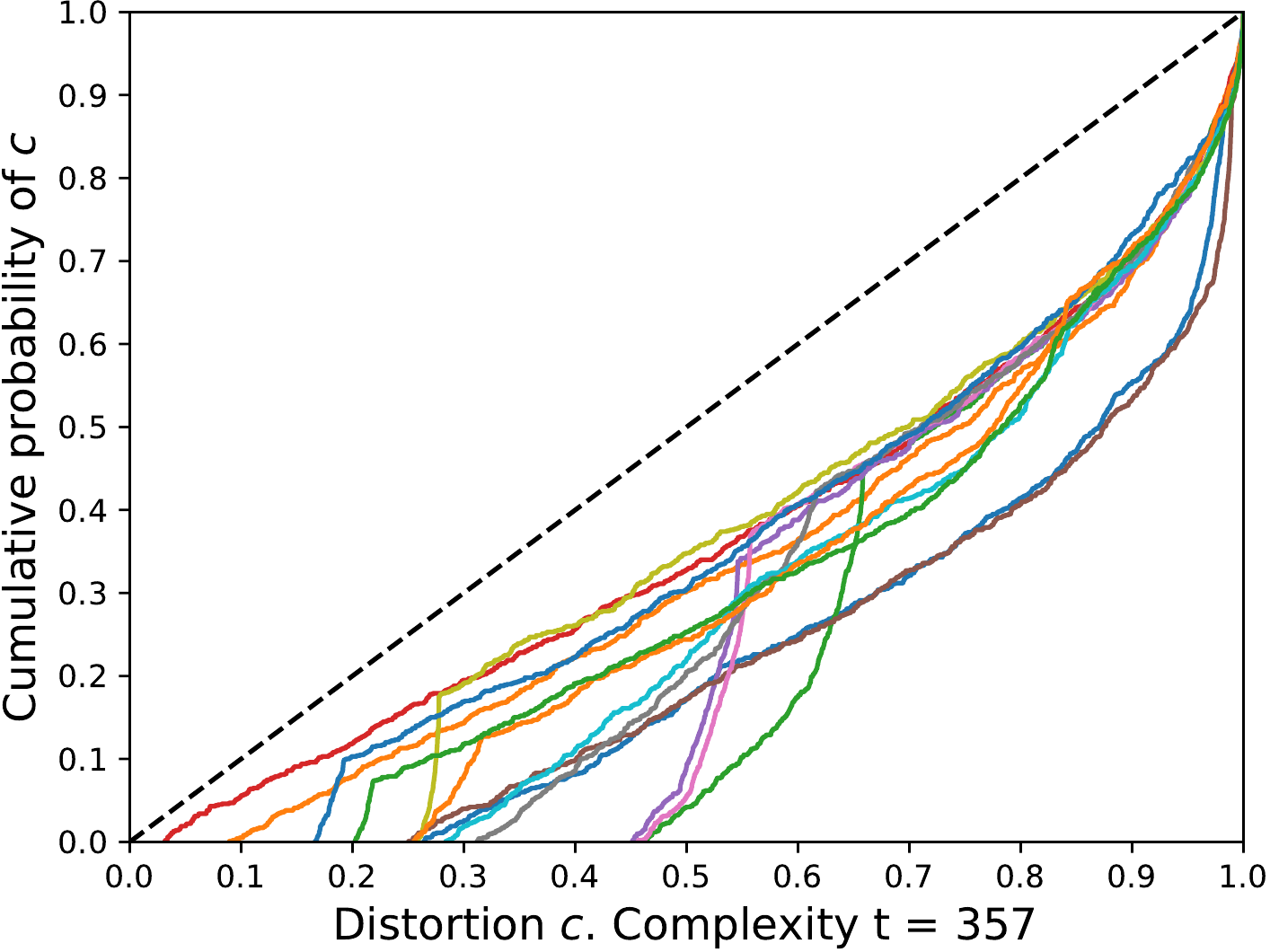}
\caption{The curves $(9392,9836)$ from the dataset $\mathcal{D}$ (left). The cumulative probability distribution of the distortion $c$, over all tested subcurves of the input pair of curves $(9392,9836)$ from $\mathcal{D}$ (right).}
\figlab{example:twocurves}
\end{figure}

E.g. we observe the pair of the curves $P$ and $Q$ (numbered 9382 and 9836) shown in \figref{example:twocurves} (left) with complexities 308 and 357 respectively. For these curves and their subcurves, the cumulative probability distributions of $c$ were calculated, over the set of results of 1000 sampled runs. We notice that the \Frechet distance of the curves $P$ and $Q$ in \figref{example:twocurves} (or their subcurves) is not dominated by one pair of vertices, and varies upon which parts of the curves are observed. For all pairs of subcurves of $P$ and $Q$ and their respective projections $P'$ and $Q'$ we may assume that for any $\gamma\in (0,1)$ it is
\begin{equation}
\text{Pr}\left[\frac{\distFr{P'}{Q'}}{\distFr{P}{Q}}\leq \gamma\right] \leq \gamma.
\label{assumptionondistribution}
\end{equation}

Indeed, when the cumulative probability distribution of the distortion $c$ is observed over all tested pairs of curves (\figref{experimentalresult1}, upper left), the mean and the standard deviation of the distortions obtained by our experiments for a given threshold $\gamma\in \lbrace 0.1, 0.2, 0.3, 0.4, 0.5, 0.6, 0.7, 0.8,$ $0.9\rbrace$, suggest that for the realistic input curves $P$ and $Q$ the assumption of Equation (\ref{assumptionondistribution}) holds with high probability. The outlying maxima occur for the curves whose shape is similar to the curves from the proof of  \thmref{lowerbounddiscretecpacked}, and thus strongly conditioned.

Furthermore, it seems that the distortion of the discrete \Frechet distance is bounded by a constant (with high probability), and that it does not depend on the complexity $t$ of the input curves, as suggested by \figref{experimentalresult1} and \figref{experimentalresult2}.

%/////////////////////////////////////////////////

\begin{figure}[ht]
\centering
\includegraphics[width=0.49\textwidth, height = 0.375\textwidth]{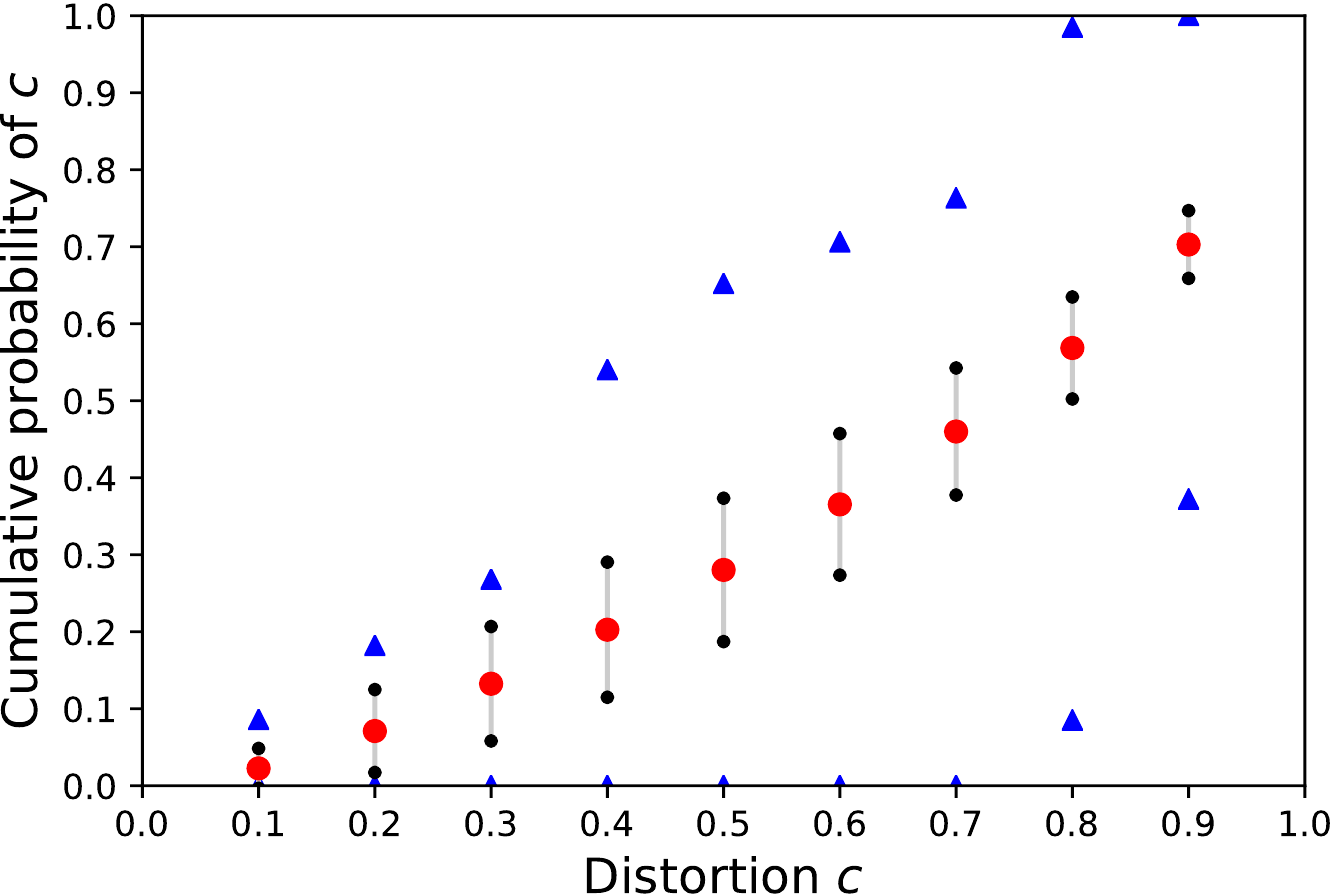}
\includegraphics[width=0.49\textwidth]{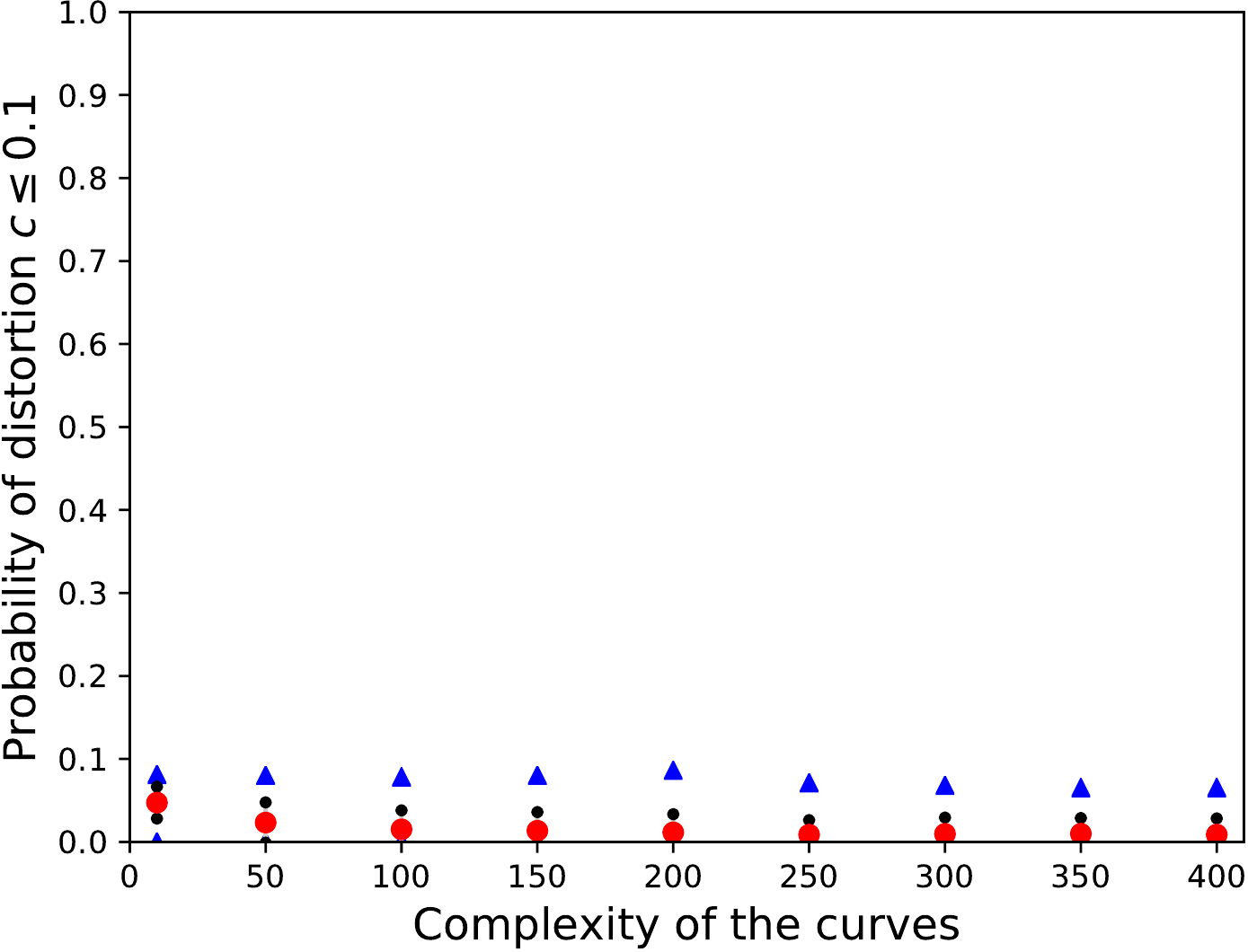}

\includegraphics[width=0.49\textwidth]{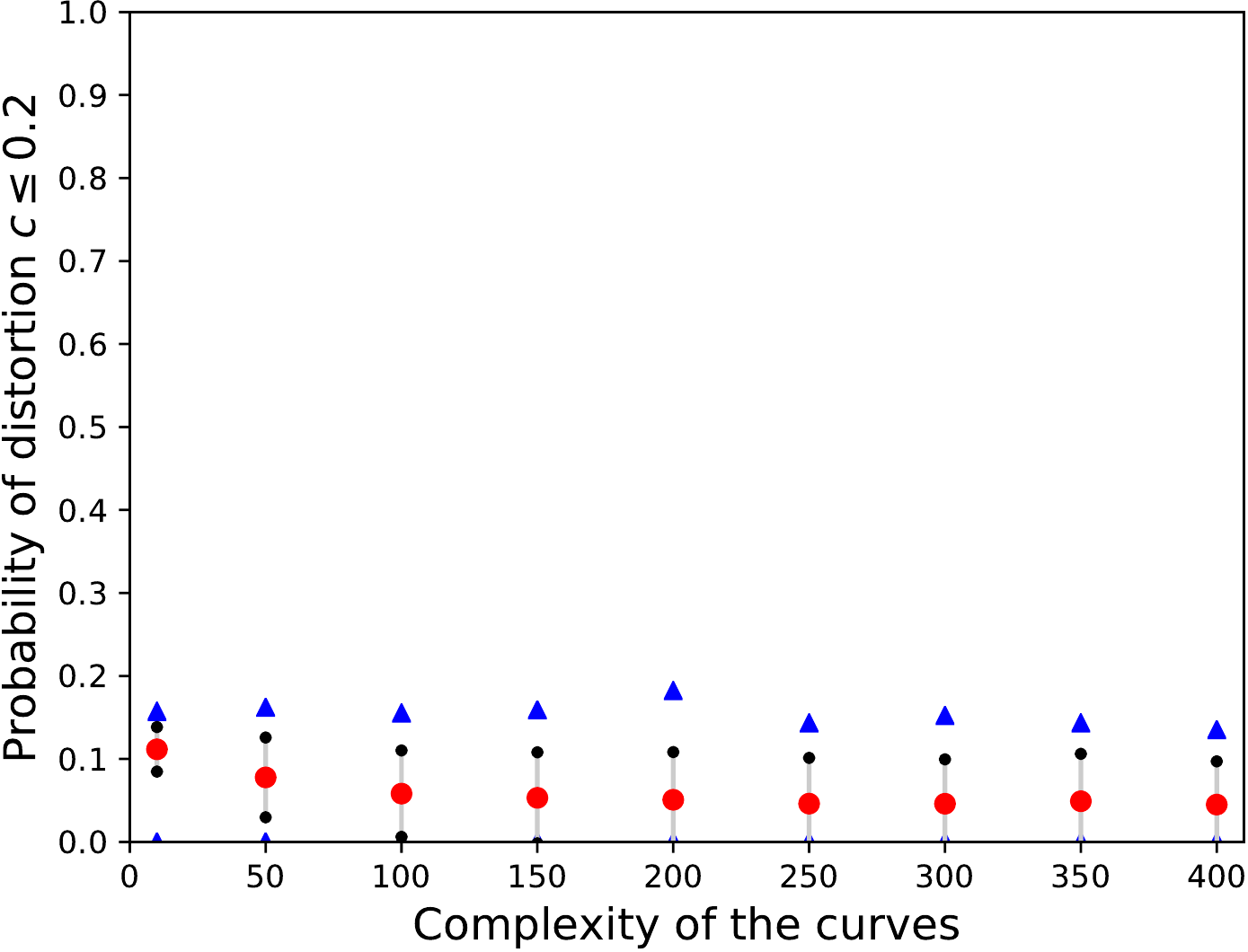}
\includegraphics[width=0.49\textwidth]{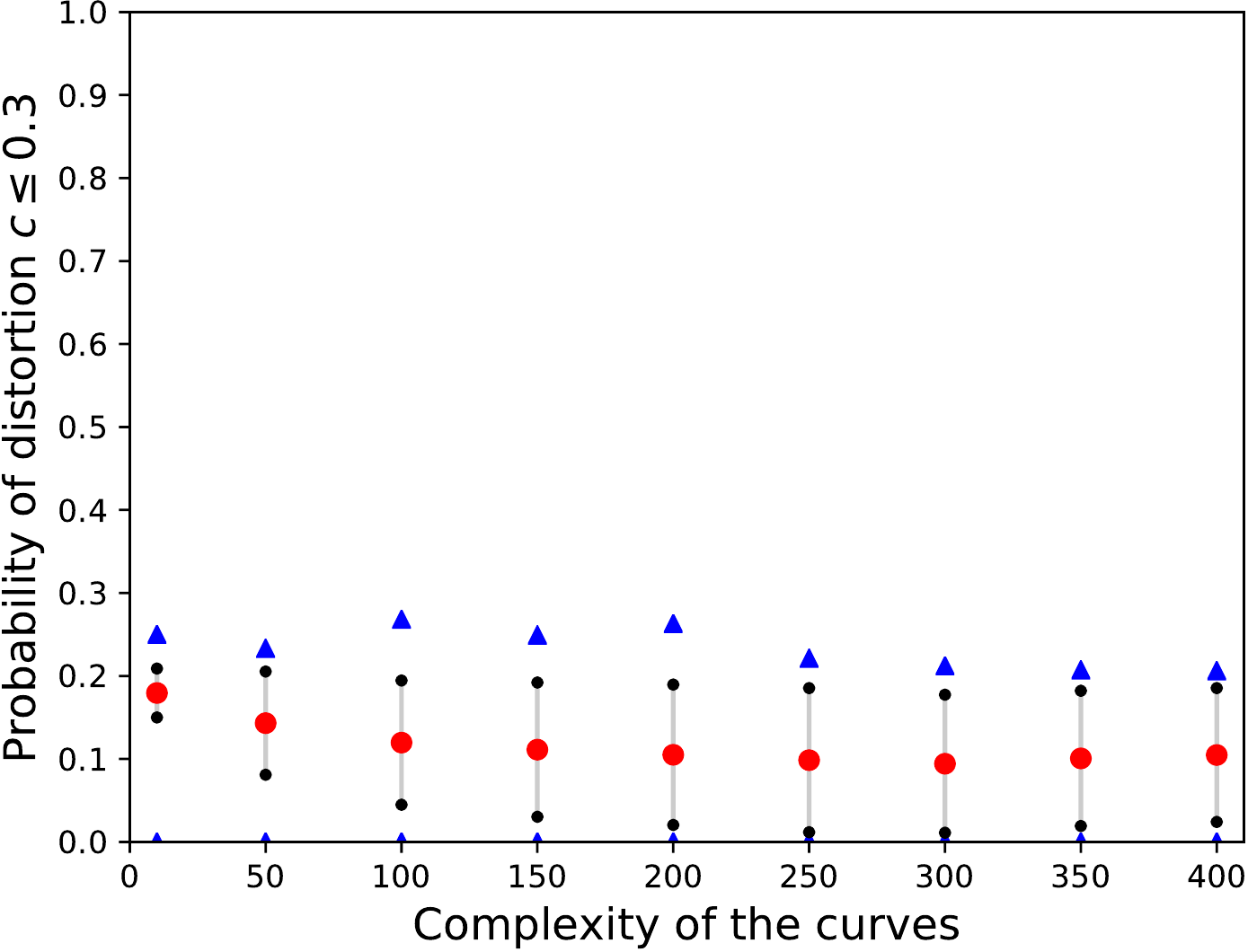}
\caption{The cumulative probability distribution of the distortion (upper left). The remaining subfigures show for a given threshold $\gamma$ of the distortion $c$, the cumulative probability $\text{Pr}[c\leq \gamma]$ as a function of the complexity $t$ of the curves, for $t\in\lbrace 10, 50, 100, 150, 200, 250, 300, 350, 400\rbrace$. %\\
The means $\mu$ of the values denoted by red circles. The intervals $[\mu-\sigma, \mu+\sigma]$ denoted by black dots, where $\sigma$ is the standard deviation. The minima and maxima denoted by blue triangles. Continued in \figref{experimentalresult2}
}
\figlab{experimentalresult1}
\end{figure}

\begin{figure}[ht]
\centering
\includegraphics[width=0.49\textwidth]{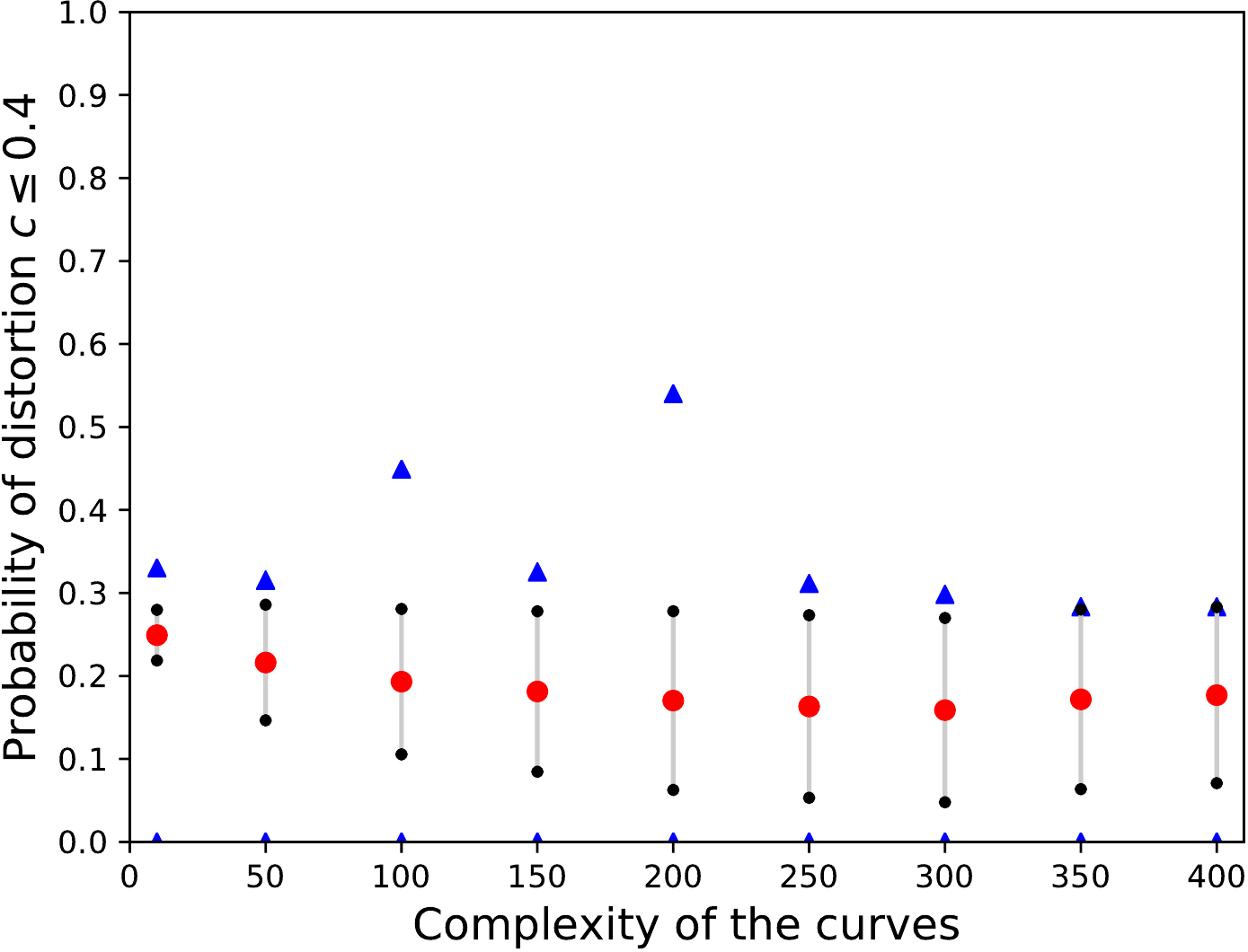}
\includegraphics[width=0.49\textwidth]{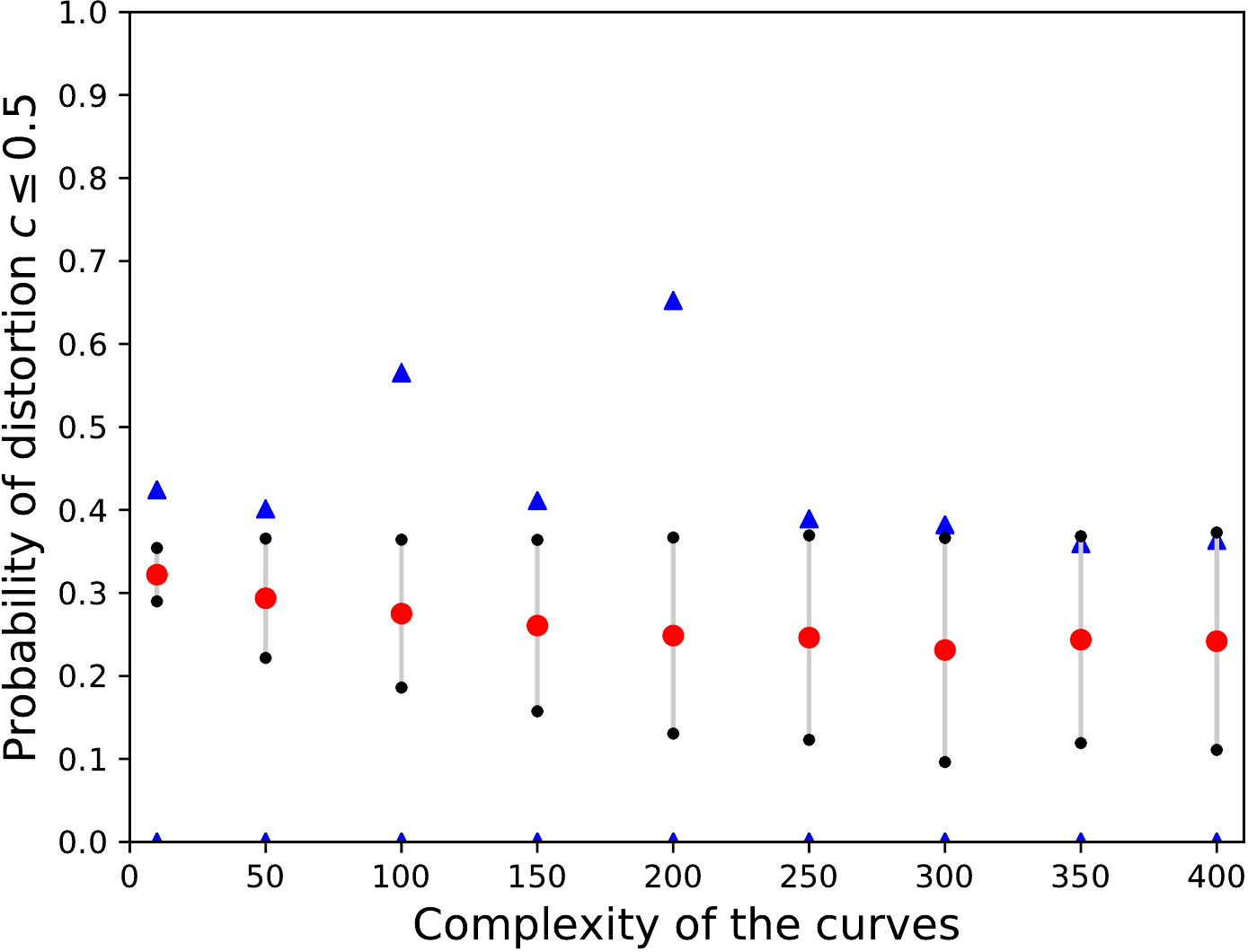}

\includegraphics[width=0.49\textwidth]{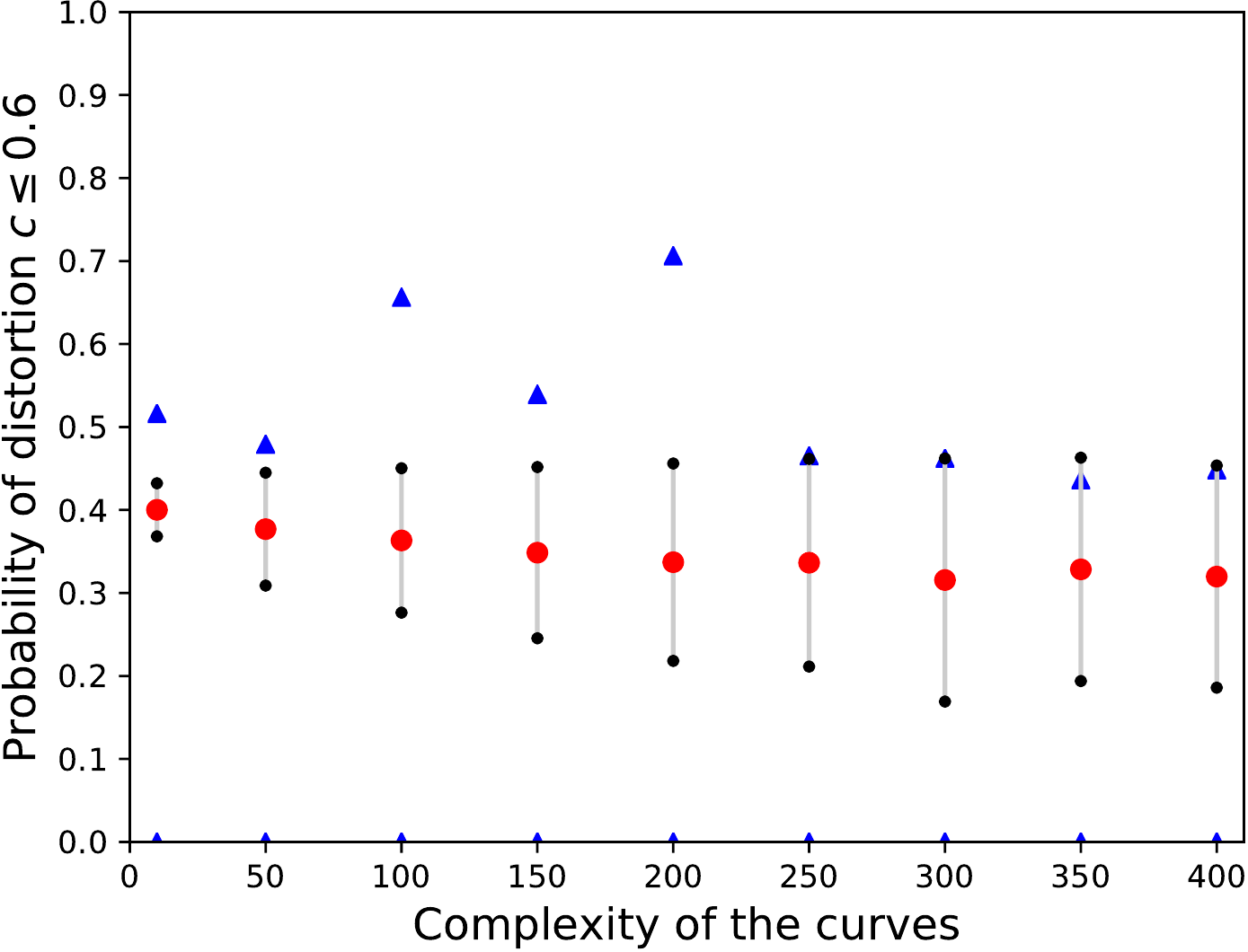}
\includegraphics[width=0.49\textwidth]{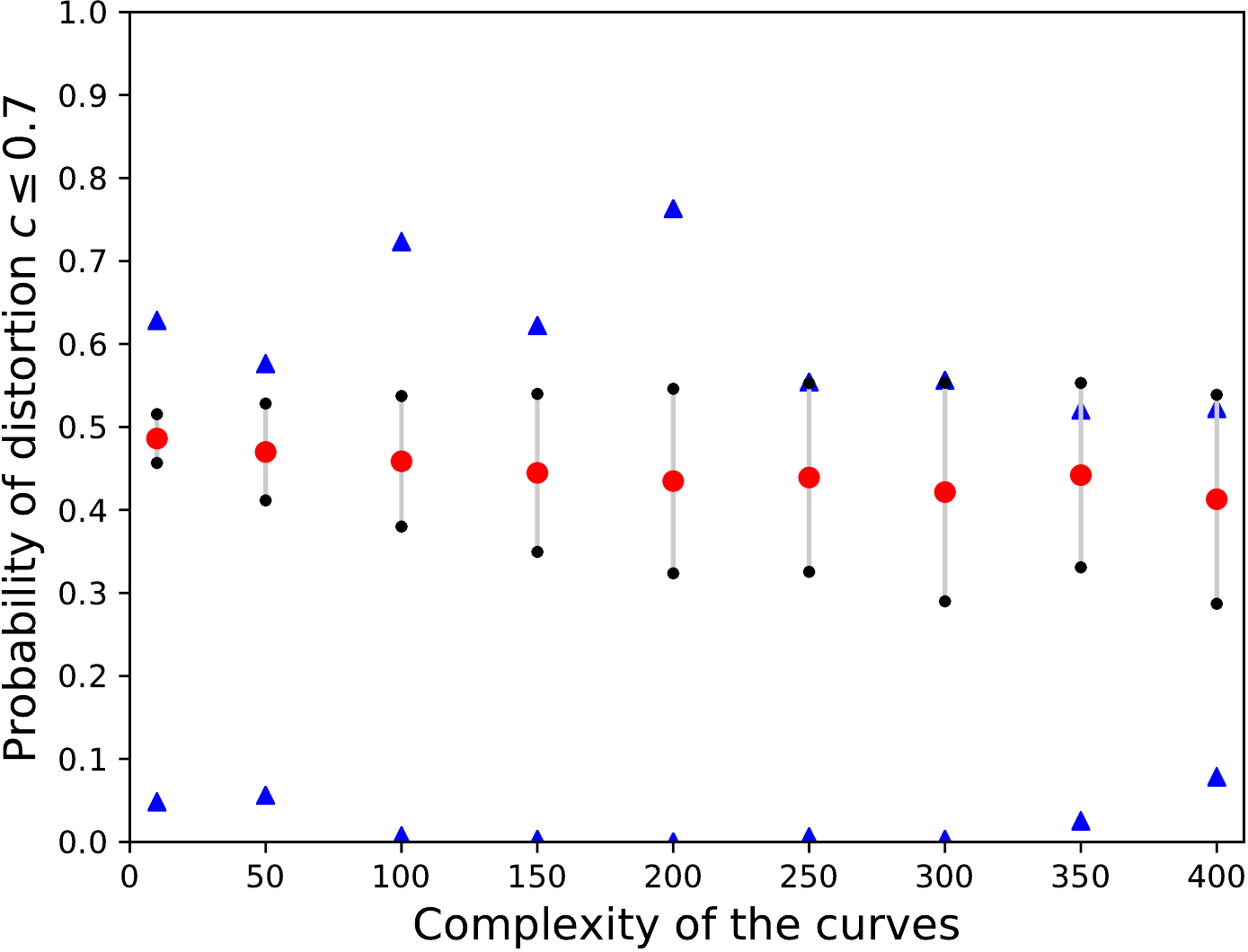}

\includegraphics[width=0.49\textwidth]{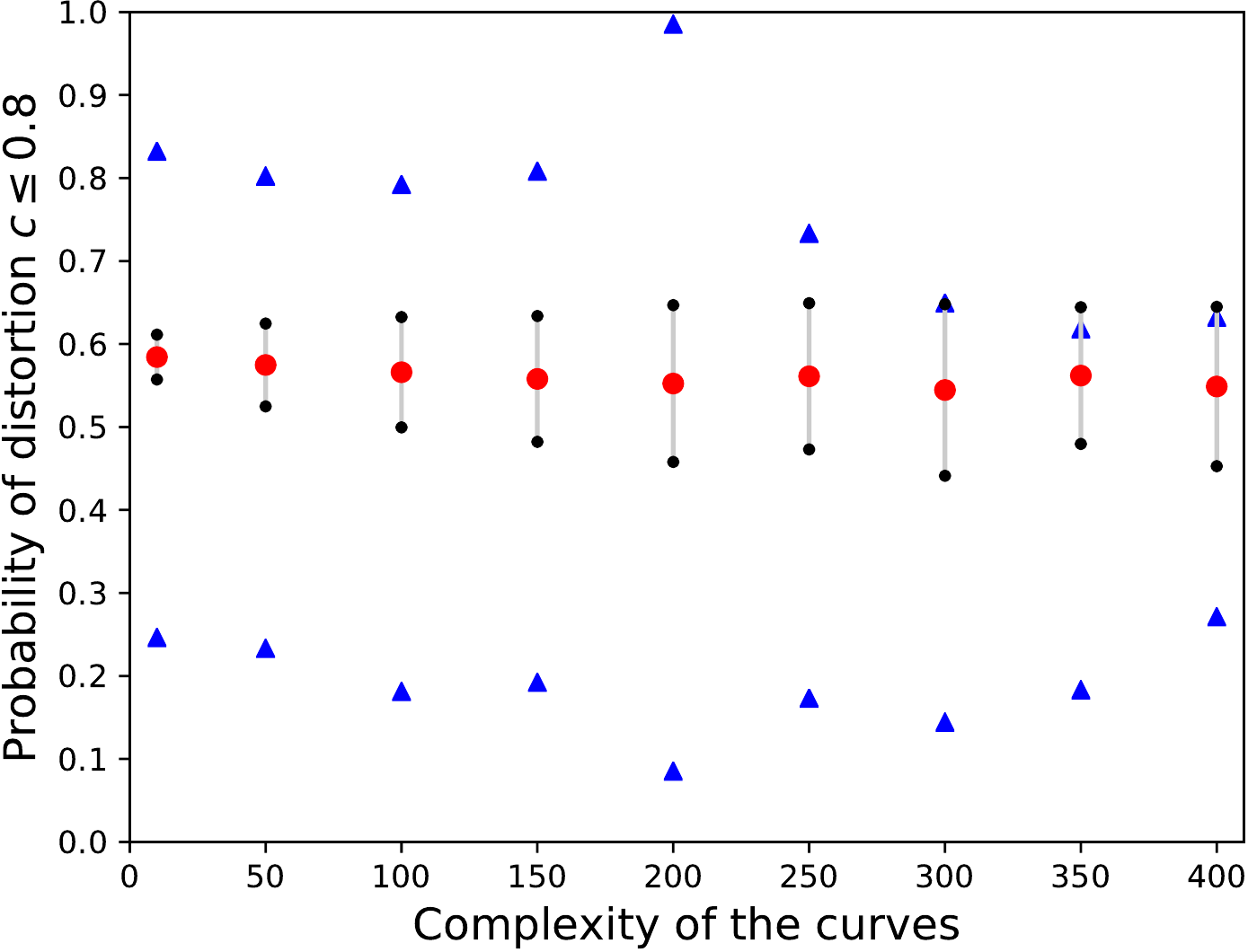}
\includegraphics[width=0.49\textwidth]{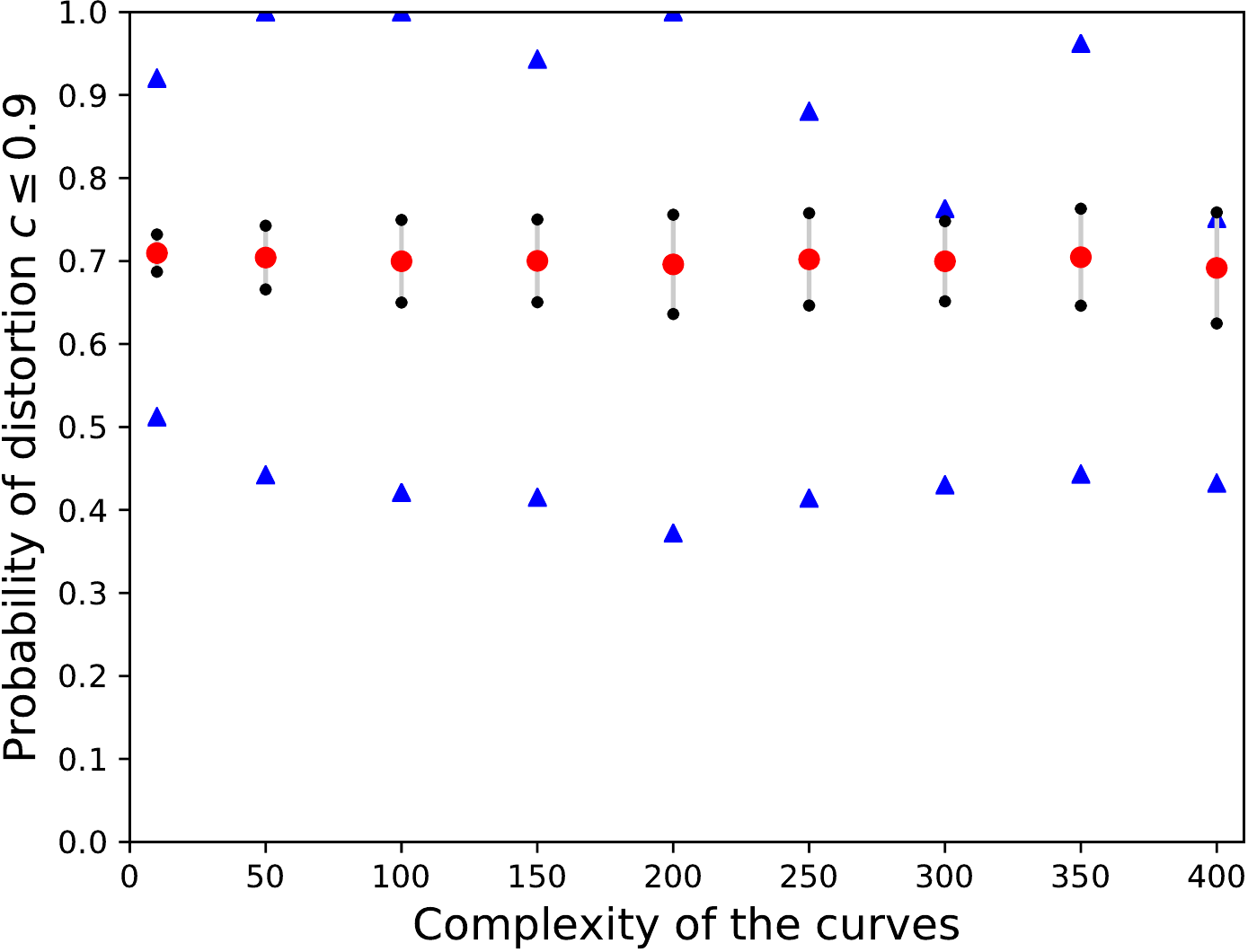}
\caption{\figref{experimentalresult1} continued.
}
\figlab{experimentalresult2}
\end{figure}

\section{Conclusions}

%\begin{figure}[t]
%\centering
%\includegraphics[height=10\baselineskip]{experiments/stat_p-crop}
%\hfill
%\includegraphics[height=10\baselineskip]{experiments/stat_t5-crop}
%\caption{The cumulative probability distribution of the distortion (left). Given $c\leq 0.5$, the cumulative probability of distortion is shown as a function of the complexity $t$ of the curves, for $t\in\lbrace 10, 50, 100, 150, 200, 250, 300, 350, 400\rbrace$ (right). %\\
%The means $\mu$ of the values denoted by red circles. The intervals $[\mu-\sigma, \mu+\sigma]$ denoted by black dots, where $\sigma$ is the standard deviation. The minima and maxima denoted by blue triangles.
%}
%\figlab{experimentalresultshort}
%\end{figure}

We studied the behavior of the discrete Fr\'echet distance between two
polygonal curves under projections to a random line.  Our results show
that in the worst case and under reasonable assumptions, the discrete Fr\'echet
distance between two polygonal curves of complexity $t$ in $\Re^d$, where $d\in\lbrace 2,3,4,5\rbrace$,
degrades by a factor linear in $t$ with constant probability. One can see this 
as a negative result, since we hoped that the Fr\'echet distance would be more 
robust under such projections. 
We also performed some preliminary experiments on the dataset of the
6th ACM SIGSPATIAL GISCUP 2017 competition (as seen in \secref{experimentalpart}).
%(see details in the appendix in \secref{experimentalpart}).
The cumulative probability distribution of the distortion\footnote{Technically speaking,
this is the inverse of the distortion as defined in the introduction. We choose this definition
to simplify the presentation, since this definition ensures that $c\in [0,1]$.}
$c=\distFr{P'}{Q'}/\distFr{P}{Q}$ (\figref{experimentalresult1}, first row, left)
suggests that for realistic input curves we can expect that $\text{Pr}\left[
c\leq \gamma\right]\leq \gamma$. This holds independently of the
complexity $t$ of the input curves, as illustrated by
\figref{experimentalresult2} (first row, right) for the given threshold $\gamma=0.5$.
This implies that with probability of at least $0.5$ we expect that the
discrete \Frechet distance will be reduced at most by a factor 2 when projected
to a line chosen uniformly at random, independently of the input complexity.
These results stand in stark contrast with our lower bounds. They indicate that
highly distorted projections happen very rarely in practice, and only for
strongly conditioned input curves.

%\subparagraph*{Acknowledgements.}
%
%We want to thank Kevin Buchin for useful discussions on the topic of this paper.

%%%Appendix.tex

%\appendix

%% file: Embedding.bbl
\begin{thebibliography}{10}

\bibitem{AbboudBW15}
A.~Abboud, A.~Backurs, and V.~V. Williams.
\newblock Quadratic-time hardness of {LCS} and other sequence similarity
  measures.
\newblock {\em CoRR}, abs/1501.07053, 2015.

\bibitem{ad-ranges-18}
P.~Afshani and A.~Driemel.
\newblock On the complexity of range searching among curves.
\newblock In {\em Proceedings of the 29th ACM-SIAM Symposium on Discrete
  Algorithms, {SODA}}, pages 898--917, 2018.

\bibitem{aaks-dfst-14}
P.~K. Agarwal, R.~{Ben Avraham}, H.~Kaplan, and M.~Sharir.
\newblock Computing the discrete {F}r{\'{e}}chet distance in subquadratic time.
\newblock {\em {SIAM} J. Comput.}, 43(2):429--449, 2014.

\bibitem{afpy-dtw-16}
P.~K. Agarwal, K.~Fox, J.~Pan, and R.~Ying.
\newblock {Approximating Dynamic Time Warping and Edit Distance for a Pair of
  Point Sequences}.
\newblock In S.~Fekete and A.~Lubiw, editors, {\em 32nd International Symposium
  on Computational Geometry, {SoCG}}, volume~51 of {\em Leibniz International
  Proceedings in Informatics (LIPIcs)}, pages 6:1--6:16, Dagstuhl, Germany,
  2016. Schloss Dagstuhl--Leibniz-Zentrum f\"ur Informatik.

\bibitem{askey2010gamma}
R.~A. Askey and R.~Roy.
\newblock Gamma function.
\newblock {\em NIST handbook of mathematical functions, US Dept. Commerce,
  Washington, DC}, pages 135--147, 2010.

\bibitem{approx_BackursS16}
A.~Backurs and A.~Sidiropoulos.
\newblock Constant-distortion embeddings of {H}ausdorff metrics into
  constant-dimensional l{\_}p spaces.
\newblock In {\em Approximation, Randomization, and Combinatorial Optimization.
  Algorithms and Techniques, {APPROX/RANDOM}}, pages 1:1--1:15, 2016.

\bibitem{BadoiuCIS05}
M.~Badoiu, J.~Chuzhoy, P.~Indyk, and A.~Sidiropoulos.
\newblock Low-distortion embeddings of general metrics into the line.
\newblock In {\em Proceedings of the 37th Annual {ACM} Symposium on Theory of
  Computing, {STOC}}, pages 225--233, 2005.

\bibitem{BadoiuDGRRRS05}
M.~Badoiu, K.~Dhamdhere, A.~Gupta, Y.~Rabinovich, H.~R{\"{a}}cke, R.~Ravi, and
  A.~Sidiropoulos.
\newblock Approximation algorithms for low-distortion embeddings into
  low-dimensional spaces.
\newblock In {\em Proceedings of the 16th Annual {ACM-SIAM} Symposium on
  Discrete Algorithms, {SODA}}, pages 119--128, 2005.

\bibitem{bartal2014impossible}
Y.~Bartal, L.~Gottlieb, and O.~Neiman.
\newblock On the impossibility of dimension reduction for doubling subsets of
  lp.
\newblock In {\em ACM Symposium on Computational Geometry, {SoCG}}, pages
  60--66, 2014.

\bibitem{Bringmann14}
K.~Bringmann.
\newblock Why walking the dog takes time: Fr\'echet distance has no strongly
  subquadratic algorithms unless {SETH} fails.
\newblock In {\em Proceedings of the 55th Annual IEEE Symposium on Foundations
  of Computer Science}, {FOCS}, pages 661--670, 2014.

\bibitem{BK-focs15}
K.~Bringmann and M.~K{\"{u}}nnemann.
\newblock Quadratic conditional lower bounds for string problems and dynamic
  time warping.
\newblock In {\em {IEEE} 56th Annual Symposium on Foundations of Computer
  Science, {FOCS}}, pages 79--97, 2015.

\bibitem{BringmannK17}
K.~Bringmann and M.~K{\"{u}}nnemann.
\newblock Improved approximation for {F}r{\'e}chet distance on $c$-packed
  curves matching conditional lower bounds.
\newblock {\em Int. J. Comput. Geom. Appl.}, 27(1-2):85--120, 2017.

\bibitem{buchin2014four}
K.~Buchin, M.~Buchin, W.~Meulemans, and W.~Mulzer.
\newblock Four {S}oviets walk the dog-with an application to {A}lt's
  conjecture.
\newblock {\em Proceedings of the 25th Annual ACM-SIAM Symposium on Discrete
  Algorithms}, pages 1399--1413, 2014.

\bibitem{folding-17}
K.~Buchin, J.~Chun, M.~L{\"{o}}ffler, A.~Markovic, W.~Meulemans, Y.~Okamoto,
  and T.~Shiitada.
\newblock Folding free-space diagrams: Computing the {F}r{\'{e}}chet distance
  between 1-dimensional curves (multimedia contribution).
\newblock In {\em 33rd International Symposium on Computational Geometry,
  {SoCG}}, pages 64:1--64:5, 2017.

\bibitem{BergCG13}
M.~de~Berg, A.~F. Cook, and J.~Gudmundsson.
\newblock Fast {F}r{\'{e}}chet queries.
\newblock {\em Comput. Geom.}, 46(6):747--755, 2013.

\bibitem{dtswk-08}
H.~Ding, G.~Trajcevski, P.~Scheuermann, X.~Wang, and E.~Keogh.
\newblock Querying and mining of time series data: Experimental comparison of
  representations and distance measures.
\newblock {\em Proc. VLDB Endow.}, 1(2):1542--1552, Aug. 2008.

\bibitem{dh-jydfd-13}
A.~Driemel and S.~{Har-Peled}.
\newblock Jaywalking your dog -- computing the {Fr\'{e}chet} distance with
  shortcuts.
\newblock {\em SIAM Journal of Computing}, 42(5):1830--1866, 2013.

\bibitem{dhw-afd-12}
A.~Driemel, S.~{Har-Peled}, and C.~Wenk.
\newblock Approximating the {F}r{\'e}chet distance for realistic curves in
  near-linear time.
\newblock {\em Discrete {\&} Computational Geometry}, 48(1):94--127, 2012.

\bibitem{DriemelKS16}
A.~Driemel, A.~Krivo\v{s}ija, and C.~Sohler.
\newblock Clustering time series under the {Fr{\'{e}}chet} distance.
\newblock In {\em Proceedings of the 27th Annual {ACM-SIAM} {S}ymposium on
  {D}iscrete {A}lgorithms, {SODA}}, pages 766--785, 2016.

\bibitem{ds-lshc-17}
A.~Driemel and F.~Silvestri.
\newblock Locally-sensitive hashing of curves.
\newblock In {\em 33st International Symposium on Computational Geometry,
  {SoCG}}, pages 37:1--37:16, 2017.

\bibitem{eiter1994computing}
T.~Eiter and H.~Mannila.
\newblock Computing discrete {F}r{\'e}chet distance.
\newblock Technical Report CD-TR 94/64, Christian Doppler Laboratory, 1994.

\bibitem{EmirisP18}
I.~Z. Emiris and I.~Psarros.
\newblock Products of {E}uclidean metrics and applications to proximity
  questions among curves.
\newblock In B.~Speckmann and C.~D. T{\'{o}}th, editors, {\em 34th
  International Symposium on Computational Geometry, {SoCG}}, pages
  37:1--37:13, 2018.

\bibitem{FellowsFLLRS13}
M.~R. Fellows, F.~V. Fomin, D.~Lokshtanov, E.~Losievskaja, F.~A. Rosamond, and
  S.~Saurabh.
\newblock Distortion is fixed parameter tractable.
\newblock {\em {TOCT}}, 5(4):16:1--16:20, 2013.

\bibitem{gs-dtw-17}
O.~Gold and M.~Sharir.
\newblock Dynamic time warping and geometric edit distance: Breaking the
  quadratic barrier.
\newblock In {\em 44th International Colloquium on Automata, Languages, and
  Programming, {ICALP}}, pages 25:1--25:14, 2017.

\bibitem{HastadIL03}
J.~H{\aa}stad, L.~Ivansson, and J.~Lagergren.
\newblock Fitting points on the real line and its application to {RH} mapping.
\newblock {\em J. Algorithms}, 49(1):42--62, 2003.

\bibitem{huber1982gamma}
G.~Huber.
\newblock Gamma function derivation of n-sphere volumes.
\newblock {\em The American Mathematical Monthly}, 89(5):301--302, 1982.

\bibitem{Indyk01}
P.~Indyk.
\newblock Algorithmic applications of low-distortion geometric embeddings.
\newblock In {\em 42nd Annual Symposium on Foundations of Computer Science,
  {FOCS}}, pages 10--33, 2001.

\bibitem{i-approxnn-02}
P.~Indyk.
\newblock Approximate nearest neighbor algorithms for {F}r\'echet distance via
  product metrics.
\newblock In {\em Symposium on Computational Geometry, {SoCG}}, pages 102--106,
  2002.

\bibitem{IndMat04}
P.~Indyk and J.~Matou\v{s}ek.
\newblock Low-distortion embeddings of finite metric spaces.
\newblock In J.~E. Goodman and J.~O'{}Rourke, editors, {\em Handbook of
  Discrete and Computational Geometry}, pages 177--196. CRC Press, 2004.

\bibitem{keogh2005exact}
E.~Keogh and C.~A. Ratanamahatana.
\newblock Exact indexing of dynamic time warping.
\newblock {\em Knowledge and information systems}, 7(3):358--386, 2005.

\bibitem{matousek1996distortion}
J.~Matou{\v{s}}ek.
\newblock On the distortion required for embedding finite metric spaces into
  normed spaces.
\newblock {\em Israel Journal of Mathematics}, 93(1):333--344, 1996.

\bibitem{mueller07dtw}
M.~M\"uller.
\newblock Dynamic time warping.
\newblock In {\em Information Retrieval for Music and Motion}, pages 69--84.
  Springer Berlin Heidelberg, 2007.

\bibitem{NayyeriR15}
A.~Nayyeri and B.~Raichel.
\newblock Reality distortion: Exact and approximate algorithms for embedding
  into the line.
\newblock In V.~Guruswami, editor, {\em {IEEE} 56th Annual Symposium on
  Foundations of Computer Science, {FOCS}}, pages 729--747, 2015.

\bibitem{RakthanmanonCMBWZZK12}
T.~Rakthanmanon, B.~J.~L. Campana, A.~Mueen, G.~E. A. P.~A. Batista, M.~B.
  Westover, Q.~Zhu, J.~Zakaria, and E.~J. Keogh.
\newblock Searching and mining trillions of time series subsequences under
  dynamic time warping.
\newblock In {\em The 18th {ACM} {SIGKDD} International Conference on Knowledge
  Discovery and Data Mining}, pages 262--270, 2012.

\end{thebibliography}
